\documentclass[letterpaper]{article}

\title{Unfairness Despite Awareness: Group-Fair Classification with Strategic Agents}
\author{
    Andrew Estornell\textsuperscript{\rm 1}, 
    Sanmay Das\textsuperscript{\rm 2}, 
    Yang Liu\textsuperscript{\rm 3}, 
    Yevgeniy Vorobeychik\textsuperscript{\rm 1}\\
    \small \textsuperscript{\rm 1} Washington University in Saint Louis, \textsuperscript{\rm 2} George Mason University,\\ \small \textsuperscript{\rm 3} University of California Santa Cruz
}
\date{}

\usepackage[utf8]{inputenc}
\usepackage{natbib}
\newcommand{\X}{\mathcal{X}}
\newcommand{\x}{\mathbf{x}}
\newcommand{\Y}{\mathcal{Y}}
\newcommand{\w}{\mathbf{w}}

\newcommand{\D}{\mathcal{D}}
\newcommand{\M}{\mathcal{M}}
\newcommand{\PR}{\text{PR}}
\newcommand{\TPR}{\text{TPR}}
\newcommand{\FPR}{\text{FPR}}
\newcommand{\Hp}{\mathcal{H}}

\newcommand{\vv}{\mathbf{v}}

\newcommand{\Pb}[1]{\mathbb{P}\big(#1\big)}

\usepackage{xcolor}
\usepackage{nicefrac, xfrac}
\usepackage{amsfonts}
\usepackage{amsmath}
\usepackage{amssymb}
\usepackage{amsthm}

\newtheorem{theorem}{Theorem}
\newtheorem{definition}{Definition}
\newtheorem{lemma}{Lemma}

\begin{document}

\maketitle

\begin{abstract}
The use of algorithmic decision making systems in domains which impact the financial, social, and political well-being of people has created a demand for these decision making systems to be ``fair'' under some accepted notion of equity. 
This demand has in turn inspired a large body of work focused on the development of fair learning algorithms which are then used in lieu of their conventional counterparts.
Most analysis of such fair algorithms proceeds from the assumption that the people affected by the algorithmic decisions are represented as immutable feature vectors. However, strategic agents may possess both the ability and the incentive to manipulate this observed feature vector in order to attain a more favorable outcome.
We explore the impact that strategic agent behavior could have on fair classifiers
and derive conditions under which this behavior leads to fair classifiers becoming less fair than their conventional counterparts under the same measure of fairness that the fair classifier takes into account. These conditions are related to the the way in which the fair classifier remedies unfairness on the original unmanipulated data: fair classifiers which remedy unfairness by becoming more selective than their conventional counterparts are the ones that become less fair than their counterparts when agents are strategic.
We further demonstrate that both the increased selectiveness of the fair classifier, and consequently the loss of fairness, arises when performing fair learning on domains in which the advantaged group is overrepresented in the region near (and on the beneficial side of) the decision boundary of conventional classifiers.
Finally, we observe experimentally, using several datasets and learning methods, that this \emph{fairness reversal} is common, and that
our theoretical characterization of the fairness reversal conditions indeed holds in most such cases.
\end{abstract}

\section{Introduction}

The increasing deployment of algorithmic decision making systems in social, political, and economic domains has brought with it a demand that fairness of decisions be a central part of algorithm design. While the specific notion of fairness appropriate to a domain is often a matter of debate, several 
have come to be commonly used in prior literature, such as positive (or selection) rate and false positive rate.
A common goal in the design of fairness-aware (\emph{group-fair}) algorithms is to balance predictive efficacy (such as accurate) with achieving near-equality on a chosen fairness measure among demographic categories, such as race or gender.
A question that arises in many domains where such ``fair'' algorithms could be used is whether they are susceptible to, and create incentives for, manipulation by agents who may misrepresent themselves in order to achieve better outcomes. 


Fair classifiers are often deployed in domains where feature manipulation is possible, under some constraints or potentially at some cost to agents. 
For example, in selection of individuals to receive assistance from social service programs (e.g homelessness services), or in admission to selective educational programs, it may be possible for applicants to misreport things, such as the number of dependents, income, or other self-reported characteristics. 
Importantly, we assume that agents cannot lie about group membership itself, and that the classifier cannot use group membership in making predictions (it can only do so during training). This is in keeping with most policy frameworks where fair machine learning might be used, since explicit use of protected categories in decision-making is often illegal.

We investigate the effects of such strategic manipulation of a binary \emph{group-fair} classifier.
In the context of the social services example, the classifier's job is to determine if an applicant should, or should not, be granted assistance, and the fairness guarantee of this classifier could be approximate equality of false positive rate between male and female applicants.
Specifically, we aim to understand the circumstances under which a group-fair classifier becomes less fair than its conventional analog (i.e., a similar classifier that does not explicitly consider group fairness).
Our main high-level observation is that if a group-fair classifier achieves greater fairness by becoming more selective than a fairness-agnostic classifier, it is also more unfair than the latter when agents are strategic.
We first establish this theoretically, albeit in restricted settings, and subsequently show that this \emph{fairness reversal} property obtains for several classifiers and datasets, and our theoretical analysis helps explain why.


\smallskip
\noindent{\bf Summary of results: }
We begin by examining threshold classifiers which operate on a single predictive feature  (or an engineered score).
Surprisingly, we show that in this setting strategic manipulation leads to a group-fair classifier being \emph{less fair than its baseline counterpart} if and only if the group-fair classifier has a higher threshold (i.e., is more selective) than the baseline classifier.
Furthermore, we provide conditions on the distribution of data for this to occur.

While the multivariate case is more complex, much of the intuition carries over. 
In particular, we demonstrate that for any hypothesis class from which the baseline and group-fair classifiers are sourced, there always exists a data distribution and a cost of manipulation such that fair classifiers become less fair under strategic manipulation. Additionally we show that when the fair classifier is more selective (aas formalized below),
there exists a cost function such that the group-fair classifier becomes less fair than the base classifier in the presence of strategic agents. 
Lastly, we experimentally evaluate two state-of-the-art group-fair learning algorithms on several standard datasets and demonstrate that in both the single-  and multi-variable case, these algorithms often  indeed produce models which are less fair than their baseline counterparts when agents act strategically.

\smallskip
\noindent{\bf Selectivity and fairness: }
Fairness reversal is a consequence of group-fair models becoming more selective than their conventional counterparts.
In the case of threshold classifiers, higher selectivity means that the group-fair classifier has a higher selection threshold $\theta_F$ than the base (accuracy-maximizing) classifier ($\theta_C$).
\begin{figure}[h]
    \label{fig:intro_example}
    \centering
    \includegraphics[scale=0.28]{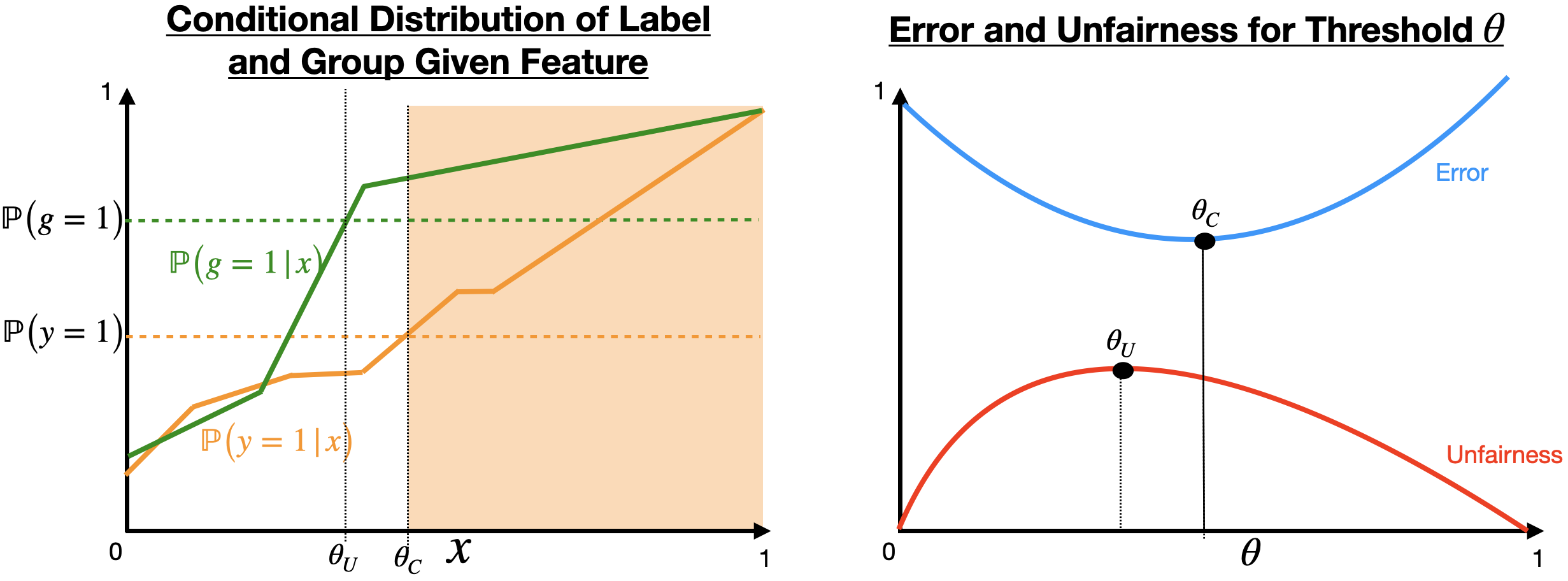}
    \includegraphics[scale=0.28]{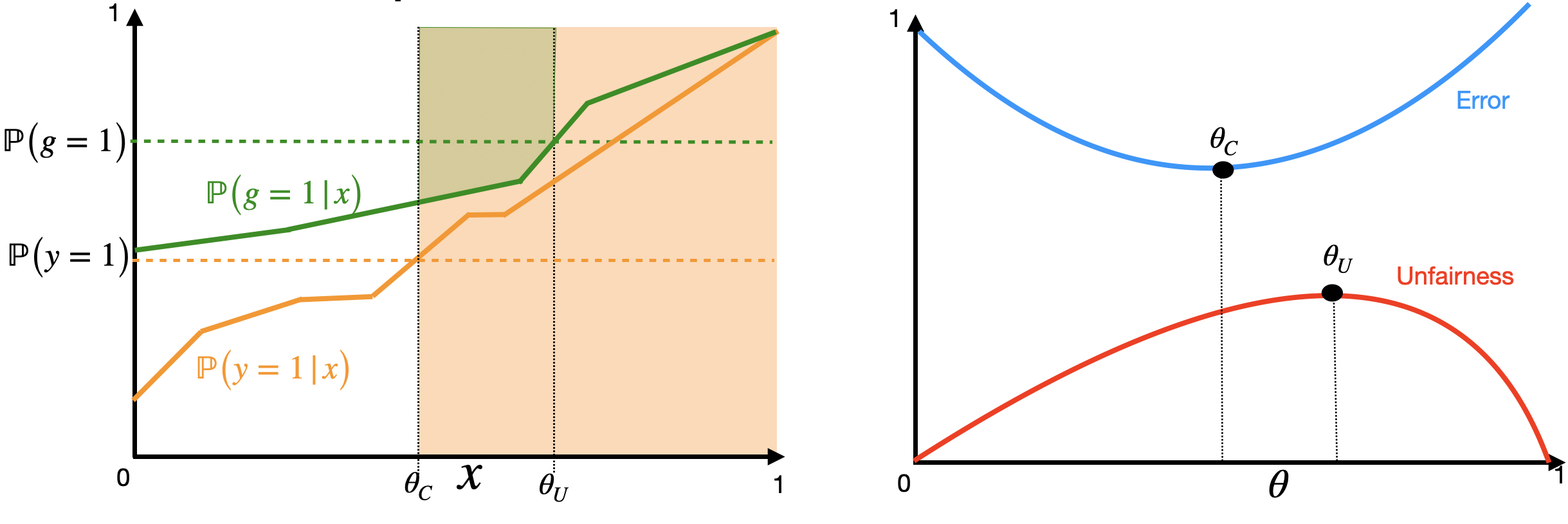}
    \caption{Stylized example of how the underlying conditional distribution leads to fair classifiers becoming more selective. The figures on the left hand correspond to the distribution, where $x$ is the predictive feature, $g$ indicates group membership, and $y$ indicates true label. The figures on the right correspond to the accuracy and unfairness of the threshold $\theta$ with respect to the distribution. Here $\theta_C$ corresponds to the threshold with the maximum accuracy, and $\theta_U$ is the threshold with the maximum unfairness.}
    \label{fig:my_label_1}
\end{figure}
When does this happen? Figure \ref{fig:my_label_1} provides some intuition. Define $\theta_U$ as the threshold that maximizes unfairness. As discussed above, $\theta_C$ is the threshold that maximizes accuracy. Assuming that the informative feature $x$ is positively correlated with the desirable class $y=1$, when $\theta_C > \theta_U$, the fair classifier will typically be more selective than the baseline, 
and the fairness reversal conditions will be met (top half of figure). This corresponds to a situation where the advantaged group is overrepresented in the region of the $x$-space just to the right of the decision boundary of $\theta_C$. Shifting this decision boundary to the right (excluding members of the advantaged group) increases fairness. However, doing so means that the newly excluded members of the advantaged group are now those who benefit most from strategic manipulation.
Conversely, if $\theta_C < \theta_U$ this will likely not be the case, as pushing $\theta_C$ to the right increases unfairness. While this is a stylized example, we show in Section \ref{S:exp} that the selectivity condition often holds using several real-world examples.

\section{Related Work}

Our work is closely related to two major strands in literature: algorithmic fairness, and in particular, approaches for group-fair classification, and adversarial machine learning (also called strategic classification).

Broadly, algorithmic fairness literature aims to study the extent to which algorithmic decisions are perceived as unfair, for example, by being inequitable to historically disadvantaged groups~\citep{Buolamwini18,CorbettDavies18,bolukbasi2016man,Ajunwa2016HiringBA}.
In particular, many approaches have been introduced, particularly in machine learning, that investigate how to balance fairness and task-related efficacy, such as accurate~\citep{Agarwal18,Feldman15,Kearns18,Xu21,Zafar19,Zemel13,hardt2016equality}.
Most of these impose hard constraints to ensure that pre-defined groups are near-equitable on some exogenously specified metric, such as selection (positive) rate~\citep{Agarwal18,Kearns18,Zafar19}, although alternative means, such as modifying the data to eliminate disparities, have also been proposed~\citep{Feldman15,chouldechova2017fair}.

The adversarial machine learning literature was initially motivated by security considerations, such as spam and malware detection~\citep{Lowd05,Huang11,Vorobeychik18}.
The primary issue of concern is that as we use machine learning techniques to identify malicious behavior, malicious actors change behavior characteristics to evade detection.
It has, however, come to have a far broader scope, encompassing robustness of machine learning techniques in computer vision as well as social applications~\citep{Goodfellow15,Hardt16,Bjorkegren20,dong2018strategic,chen2020learning}.
In the latter context, this has come to be known as \emph{strategic classification}, to indicate the concern that individuals impacted by algorithmic decisions change their features (e.g., by misreporting their household characteristics on surveys used to allocate housing to the homeless) and thereby undermine algorithms' efficacy.
The intersection between strategic classification and fairness is particularly salient to our work, and has featured studies that highlight the inequity that results from strategic behavior by individuals~\citep{Hu19}, as well as inequity (social cost) resulting from making classifiers robust to strategic behavior~\citep{Milli18,Xu21}.
Our goal, however, is quite distinct: we investigate the extent to which \emph{group-fair} classification itself leads to greater inequity compared to baseline approaches that do not include group-fairness constraints as a result of strategic behavior by individuals.

\section{Preliminaries}

We consider a setting with a population of agents, with each characterized by 1) a feature vector $\x \in \X$, 2) a group $g \in G \equiv \{0,1\}$ to which it belongs (as is common in much prior literature, we treat it as binary here), and 3) a (true) binary label $y \in \Y \equiv \{0,1\}$, denoting, for example, the agent's qualification (for a service, employment, bail, etc).
Let $\D$ be the joint distribution over $G\times\X\times\Y$.
We define $h(\x)$ as the \emph{marginal} pdf of $\x$, and assume that $h(\x) > 0$ for each $\x \in \X$.

Since using the sensitive group membership feature may pose a legal challenge, we assume that neither the conventional nor the group-fair classifier do so at prediction time (but may at training time).
We denote the baseline, or conventional, classifier by $f_C$, while the group-fair classifier is denoted by $f_F$, and both map feature vectors $\x$ into a binary label $y \in \Y$.
We assume that the baseline classifier aims to maximize accuracy, i.e., 
$f_C \in \arg\max_{f} \mathbb{P}_{(\x, y)}\big(f(\x) = y\big)$,
while $f_F$ aims to balance accuracy and fairness, solving
    \begin{align*}
        f_F = \arg&\text{max}_{f}~(1-\alpha) \mathbb{P}_{(\x, y)}(f(\x) = y)\\
        & - \alpha\big|\M(f_F;g=0) - \M(f_F;g=1)\big|,
    \end{align*}
where $\alpha \in [0,1]$ specifies the relative weight of accuracy and fairness terms, while $\M(f;g)$ is a measure of efficacy (e.g., positive rate) of $f$ restricted to a group $g$.

In the literature fairness is sometimes defined with hard constraints, rather than the soft constraints of $\alpha$-fairness, for example
\begin{align*}
    f_F = \arg&\text{max}_{f}~\mathbb{P}_{(\x, y)}(f(\x) = y)\\
    & \text{subj. to } \big|\M(f_F;g=0) - \M(f_F;g=1)\big| \leq \beta
\end{align*}
With hard constraints, decreasing $\beta$ can never increase the unfairness of $f_F$. In general soft constraints do not have the propriety that increasing $\alpha$ will never increase the unfairness of $f_F$. However, in the settings we study this is not an issue as there is a direct correspondence between $\alpha$ and $\beta$ fairness.
    
We consider the impact of strategic behavior of agents when they face a classifier $f$ (whether baseline or group-fair).
Specifically, we suppose that each agent with features $\x$ can modify these, transforming them into another feature vector $\x'$ that is reported to the classifier.
In doing so, the agent incurs a cost, captured by a cost function $c(\x,\x') \ge 0$.
As in \citet{Hardt16}, we define the agent's utility to be 
\[
u(\x,\x') = f(\x') - f(\x) - c(\x, \x').
\] Following the standard setting in strategic classification or adversarial machine learning, we assume any misreporting behavior would not change the true nature of $\x$'s label $y$.
We assume that all agents are rational utility maximizers.
Thus, since $f(\x') - f(\x) \le 1$, the agent will misreport its features only when $c(\x,\x') \le 1$; additionally, the agent will not misreport if $f(\x) = 1$ (they are selected even with true values of features).

\section{Classifiers on a Single Variable}\label{sec:single_var}
First we consider the case in which there is a single continuous feature, i.e. $\X = [0, 1]$ and classifiers are thresholds on this feature.
Recall that $f_C$ is the base classifier selected for maximum accuracy, and $f_F$ is an $\alpha$-fair classifier w.r.t. to fairness metric $\M$.
We can express both classifiers as single parameter $\theta_C, \theta_F\in[0, 1]$ respectively where
$f(x) = \mathbb{I}[x \geq \theta]$. We start by formalizing some standard notions we use extensively.

\subsection{Preliminaries and preparations}
\begin{definition}
\textbf{(Unimodal):} A function $H:[a, b] \rightarrow \mathbb{R}$ is \emph{negatively unimodal} (\emph{positively unimodal}) on the interval $[a, b]$ if there exists a point ${r\in[a, b]}$ such that $H$ is monotone decreasing (increasing) on $[a, r]$ and monotone increasing (decreasing) on $[r, b]$. \\(All convex functions are \emph{negatively unimodal} and all concave functions are \emph{positively unimodal}.)
\end{definition}

\begin{definition}\label{def:single_cross}
\textbf{(Single Crossing):} A function $f$ is said to have a single crossing with the function $g$ if there exists $z$ s.t. 
\begin{align*}
     \forall x \leq z: f(x) \geq g(x)~\text{ and }~\forall x \geq z: f(x) \geq g(x)
\end{align*}
\end{definition}
This \emph{single crossing} property is relevant to our work as we look at conditional distributions which have a single crossing with their unconditioned counterpart, e.g.. the functions $\Pb{y=1|x}$ and  $\Pb{y=1}$ have a single crossing w.r.t. $x$. 
Experimentally we observe that this single crossing between $\Pb{y=1|x}$ and $\Pb{y=1}$, as well as $\Pb{g=1|x}$ and $\Pb{g=1}$ where $g$ indicates group membership, almost always holds in practice.
Note that any monotone function has a single crossing with \emph{every} constant function and thus monotone conditionals trivially satisfy this condition. 

We begin our investigation by proving several lemmas (whose proofs are provided in the supplement) which will be useful in proving our main results, namely Theorems \ref{thm:theta_C<theta_F} and \ref{thm:x_g<x_y}, later in the section.


\begin{lemma}\label{lem:error_convex}
    Suppose that $\mathbb{P}\big(y=1| x)$ has a single crossing with $\mathbb{P}(y=1)$. 
    Then error is negatively unimodal w.r.t. $\theta$ and the optimal base threshold is $\theta_C$ s.t. $\Pb{y=1|\theta_C} = \Pb{y=1}$.
\end{lemma}

\begin{definition}
    \textbf{(PR, TPR, FPR):} Postive Rate (PR), True Positive Rate (TPR), and False Positive Rate (FPR) are defined, for classifier $f$ and distribution $\D$ over $G \times \X \times \Y$, as
    \begin{align*}
        \PR_{\D}(f)  &= \Pb{f(x) = 1}\\
        \TPR_{\D}(f) &= \Pb{f(x) = 1|y=1}\\
        \FPR_{\D}(f) &= \Pb{f(x) = 1|y=0}
    \end{align*}
\end{definition}

\begin{lemma}\label{lem:PR_concave}
   Suppose that fairness is defined in terms of Positive Rate $(\PR)$ and that ${\Pb{g=1|x}}$ has a single crossing with  $\Pb{g=1}$, then
    \begin{enumerate}
        \item $\PR_{\D}(\theta|g=1) \geq \PR_{\D}(\theta|g=0)$ for any $\theta \in [0, 1]$, (i.e. group 1 is advantaged under any threshold classifier), and
        \item the unfairness term $\big|\PR_{\D}(\theta|g=1) - \PR_{\D}(\theta|g=0)\big|$ is \emph{positively unimodal} w.r.t. $\theta$ and is maximized at any $\theta_U$ s.t. $\Pb{g=1|x=\theta_U} = \Pb{g=1}$.
    \end{enumerate}
\end{lemma}

\begin{lemma}\label{lem:TPR_FPR_concave}
    Suppose that fairness is defined by either True Positive Rate or False Positive Rate and that $g,y$ are conditionally independent given $x$. Suppose further that ${\Pb{g=1|x}}$ has a single crossing with $\Pb{g=1|y=1}$ in the TPR case and by $\Pb{g=1|y=0}$ in the FPR case. Then when $\M$ is TPR or FPR,  
    \begin{enumerate}
        \item $\M_{\D}(\theta|g=1) \geq \M_{\D}(\theta|g=0)$  for any $\theta \in [0, 1]$, (i.e. group 1 is advantaged under any threshold classifier), and
        \item the unfairness term $\big|\M_{\D}(\theta|g=1) - \M_{\D}(\theta|g=0)\big|$ is \emph{positively unimodal} w.r.t. $\theta$ and is maximized at any $\theta_U$ s.t. $\Pb{g=1|x=\theta_U} = \Pb{g=1|y=1}$ in the TPR case and $\Pb{g=1|x=\theta_U} = \Pb{g=1|y=0}$ in the FPR case.
    \end{enumerate}
\end{lemma}


\begin{lemma}\label{lem:thresh_same_as_strat_behave}
    Suppose $f$ is of the form $f(x) = \mathbb{I}[x \geq \theta]$, and the cost of manipulating a feature from $x$ to $x'$ is given as $c(x, x')$, where an agent with true feature $x$ can submit any $x'$ subject to $c(x, x') \leq B$. 
    Then for any cost function $c$ which is monotone in $|x' - x|$ there exists a classifier $f'(x) = \mathbb{I}[x \geq \theta']$ which makes identical predictions on the true distribution $\D$ as $f$ makes on manipulated data $\D_{f}^{(B)}$, i.e. when agents behave strategically $f(x') = f'(x)$ for all $x\in \X$.
    Moreover 
    \begin{align*}
        \theta' = \text{argmin}_{x}x \quad~\text{ s.t. }~ c(x, \theta) \leq B.
    \end{align*}
\end{lemma}


Lemma \ref{lem:thresh_same_as_strat_behave} implies that strategic agent behavior can be examined through both the perspective of the original classifier $f$ predicting on the modified distribution $\D_{f}^{(B)}$ or a modified classifier $f'$ on the original distribution $\D$. Since our investigation involves comparing two classifiers, $f_C$ and $f_F$, the latter perspective will prove particularly useful given that the distribution $\D$ remains invariant between classifiers.

\subsection{Fair vs. baseline classifiers when agents are strategic}
We now have the tools to  
compare the relative fairness of the fair classifier $f_F$ and the baseline classifier $f_C$ when agents are strategic.
First we state a necessary and sufficient condition for which the fair classifier becomes less fair than the baseline classifiers when agents best respond to either classifier, namely that $f_F$ becomes less fair than $f_C$ if and only if $f_F$ is more selective than $f_C$ (i.e. the set of examples which $f_F$ classifies as a $1$, is a subset of those classified as a $1$ by $f_C$). 
Next we state a sufficient condition on the underlying distribution $\D$ that guarantees that $f_F$ will become less fair than $f_C$ when both are learned on $\D$. We observe that this sufficient condition is frequently satisfied in our experiments. 

Note that if the budget of agents is so high that both $f_C$ and $f_F$ classify all agents as positive, then both classifiers have equal fairness and our results hold in an uninteresting manner. 
Thus we assume for the remainder of the paper that the available budgets do not induce such degenerate outcomes.
\begin{theorem}\label{thm:theta_C<theta_F}
    Suppose fairness is defined by PR, TPR, or FPR. Suppose further that $\Pb{y=1|x}$ has a single crossing (Def \ref{def:single_cross}) with $\mathbb{P}(y=1)$, ${\Pb{g=1|x}}$ has a single crossing with value given in Lemmas \ref{lem:PR_concave} and \ref{lem:TPR_FPR_concave}, $c(x, x')$ is monotone in $|x'-x|$, and $\theta_C$, $\theta_F$ are respectively the most accurate and optimal $\alpha$-fair thresholds. Then there exists a range of budgets $[B_1, B_2]$ such that strategic behavior agent behavior, with budget $B\in [B_1, B_2]$, leads to $f_F$ being less fair than $f_C$ if and only if $\theta_C < \theta_F$, (i.e. fair classifiers which are more \emph{selective} than their baseline counterpart become less fair under strategic manipulation)
\end{theorem}
\begin{proof}[Proof Sketch]
    The full proof is provided in the supplement.
    The unfairness of threshold $\theta$ w.r.t. to the distribution $\D$ and fairness metric $\M\in\{\PR, \TPR,\FPR\}$ is expressed as,
    \begin{align*}
        U_{\D}(\theta) = \big| \M_{\D}(\theta|g=1) - \M_{\D}(\theta|g=0)\big|,
    \end{align*}
    For a given threshold $\theta$ and manipulation budget $B$ the best response of an agent with true feature $x$ is
    \begin{align*}
        x_{\theta}^{(B)} = \text{argmax}_{x'}\big(\mathbb{I}[x' \geq \theta] - \mathbb{I}[x \geq \theta]\big)~\text{ s.t. } c(x, x') \leq B,
    \end{align*}
    When agents from $\D$ play this optimal response, let the resulting distribution be $\D^{(B)}_{\theta}$.
    The difference in unfairness, between classifiers, when agents are strategic is $U_{\D^{(B)}_{\theta_C}}(\theta_C)) - U_{\D^{(B)}_{{\theta_F}}}(\theta_F)$.
    By lemma \ref{lem:thresh_same_as_strat_behave} this unfairness can be expressed in terms of the true distribution $\D$, namely
    \begin{align*}
        U_{\D^{(B)}_{\theta_C}}(\theta_C)& - U_{\D^{(B)}_{\theta_F}}(\theta_F)
        =U_{\D}(\theta_C^{(B)}) - U_{\D}(\theta_F^{(B)})\\
        \text{where }&\theta_C^{(B)} = \text{argmin}_{x}x \text{ s.t. } c(x, \theta_C) \leq B \text{ and,}\\
        &\theta_F^{(B)} = \text{argmin}_{x}x \text{ s.t. } c(x, \theta_F) \leq B
    \end{align*}
  By  the monotonicity of $c$ both $\theta_C, \theta_F$ are monotonically decreasing w.r.t. $B$ and $\theta_C^{(B)} \leq \theta_F^{(B)}$.
    By the unimodality of $U_{\D}(\theta)$, 
    if there exists a $B'$ s.t. $\theta_F^{(B')} = \theta_U$, where $\theta_U = \text{argmax}_{\theta}U_{\D}(\theta)$, then for any $B>B'$ we have $U_{\D}(\theta_C^{(B)}) \leq U_{\D}(\theta_F^{(B)})$. This must be true since the only setting in which this does not occur is $\theta_C < \theta_F < \theta_U$, which would imply that $f_C$ is at least as fair as $f_F$ on $\D$. Thus the forward direction holds.
    
    Now assume that there exists a budget $B'$ such that $U_{\D}(\theta_C^{(B)}) \leq U_{\D}(\theta_F^{(B)})$. 
    In this case, if $\theta_F < \theta_C$, then the unimodality of $U_{\D}(\theta)$ implies that $\theta_U < \theta_F < \theta_U$, which would imply again that $f_C$ is at least as fair as $f_F$ on the true distribution. 
    Thus the converse direction also holds.
\end{proof}

Theorem \ref{thm:theta_C<theta_F} indicates that the fair classifier becomes less fair than the base classifier (under PR, TPR, or FPR) iff $\theta_C < \theta_F$, i.e. only fair classifiers which are more \emph{selective} than their baseline counter parts lead to greater unfairness under strategic behavior.
With this observation, the question then becomes: under what conditions is $f_F$ actually more \emph{selective} than $f_C$?
We now provide a sufficient condition on the underlying distribution $\D$ such that the optimal ${\alpha\text{-fair}}$ classifier is indeed more selective. In section \ref{S:exp} we demonstrate empirically that these sufficient conditions do in fact hold frequently in practice. 

\begin{theorem}\label{thm:x_g<x_y}
    Suppose fairness is defined by PR, TPR, or FPR. Suppose further that $\Pb{y=1|x}$ has a single crossing with $\mathbb{P}(y=1)$, and that ${\Pb{g=1|x}}$ has a single crossing with the respective value given in Lemmas \ref{lem:PR_concave} and \ref{lem:TPR_FPR_concave}, call this value $p_g$. Let $x_y$ and $x_g$ be defined by 
    \begin{align*}
        \Pb{y=1|x_y} = \mathbb{P}(y=1)~\text{ and }~\Pb{g=1|x_g} = p_g
    \end{align*}
    If $x_g < x_y$, then there exists a nonempty interval $[\alpha_0, \alpha_1]$ s.t. for any $\alpha\in [\alpha_0, \alpha_1]$ the optimal $\alpha$-fair classifier $f_F$, has the propriety that $\theta_C < \theta_F$ (implying that strategic agent behavior leads to $f_F$ becoming less fair than $f_C$ as outlined by Theorem \ref{thm:theta_C<theta_F}).
\end{theorem}
Intuitively this condition is saying that if the advantaged group (group 1) is overrepresented among features $x$ which have a slightly higher than base rate probability to be true positive examples (y=1), the optimal $\alpha$-fair classifier achieves its fairness by negatively classifying those ``borderline" features on  group 1 is overrepresented. This selectivity in turn leads to strategic agent behavior reversing the relative fairness of $f_C$ and $f_F$ since those newly rejected members of group 1 are now those who will most benefit  from strategic manipulation.
The size of this range of fairness coefficients  will be a function of how great the over-representation of group 1 is on the features $x$. (See Figure \ref{fig:my_label_1})
\begin{proof}[Proof Sketch]
    The full proof is provided in the supplement.
    Note that $x_y = \theta_C$ and $x_g = \theta_U$, and thus $\theta_U < \theta_C$.
    For metric $\M\in\{\PR, \TPR,\FPR\}$ the fair learning objective is
    \begin{align*}
        &(1-\alpha)\Pb{\mathbb{I}[x \geq \theta] = y} + \alpha U(\theta, \D)
    \end{align*}
    By Lemmas \ref{lem:error_convex}, \ref{lem:PR_concave}, \ref{lem:TPR_FPR_concave} the error term $\Pb{\mathbb{I}[x \geq \theta_F] = y}$ is negatively unimodal w.r.t. $\theta$ and the unfairness term $U(\theta, \D)$ is positively unimodal.  
    Therefore it cannot be the case that $\theta_U < \theta_F < \theta_C$, which implies that  $\theta_F$ lives either on  $[0, \theta_U)$ or $(\theta_C, 1]$. 
    Further the optimal $\theta_F$ restricted to $[0, \theta_U)$ is monotonically decreasing, and the optimal $\theta_F$ restricted to $(\theta_C, 1]$ is monotonically increasing, w.r.t. $\alpha$. Thus unfairness is monotonically decreasing on both intervals w.r.t. to $\alpha$.
    
    Since $\theta_U < \theta_C$ and both $\Pb{\mathbb{I}[x \geq \theta] = y}$ and $U(\theta, \D)$ are smooth it must be the case that there exists a $\theta'$ s.t. any $\theta\in[\theta_C, \theta']$ has the property that
\[
        \Pb{\mathbb{I}[x \geq \theta'] = y} <  \Pb{\mathbb{I}[x \geq \theta_U] = y} \text{ and } \]
        \[ U(\theta', \D) < U(\theta_U, \D)
\]
    Thus for $\alpha\in (0, \alpha']$, where $\alpha'$ is the fairness coefficient corresponding $\theta'$, the optimal fair classifier has $\theta_C < \theta_F$.
\end{proof}

We now turn our attention to a complementary observation. Previously we examined how strategic agent behavior can lead to a \emph{fairness reversal} between the fair classifier $f_F$ and the baseline classifier $f_C$. 
This fairness reversal also comes with an interesting \emph{accuracy reversal}. Not only is it the case that strategic behavior leads to $f_F$ taking on some of the unfairness of $f_C$, but also leads to $f_F$ obtaining some of the accuracy of $f_C$ as well.
This is primarily due to the fact that $f_F$ becomes more selective and therefore more resilient to manipulation. 
This benefit overrides the accuracy drop on the original distribution when the budget $B$ is sufficiently large. To illustrate this concept we look at linearly separable data, i.e. there exists a $\theta^*$ such that for $x \geq \theta^*$, $y=1$; otherwise $y=0$.

Denote the following set:
\begin{align}
    \X(\theta, B):= \{x: x < \theta, \exists x'\geq \theta~s.t.~c(x,x') \leq B\}
\end{align}
Note that for a fixed $\theta$, $\X(\theta, B)$ is monotone in $B$, i.e. a larger $B$ incurs a larger $\X(\theta,B)$.

\begin{theorem}\label{thm:acc_reversal}
    When $\theta_F>\theta_C$ (the fair classifier is more selective), and $B$ is large enough such that
    \[
    \underbrace{\Pb{x \in \X(\theta_C, B)}+\Pb{x \in \X(\theta_F, B)}}_{\text{Possible manipulations}} \geq \underbrace{\Pb{x \in [\theta_C, \theta_F]}}_{\text{Accuracy gap}}
    \]
    Then $\theta_F$ is more accurate on $\D^{(B)}_{\theta_F}$ than $\theta_C$ on $\D^{(B)}_{\theta_C}$.
\end{theorem}
\begin{proof}
   The full proof is provided in the supplement.
\end{proof}

\subsubsection{Fair classifiers when agents are strategic vs. when agents are truthful:}
Thus far we have considered fairness reversals between $f_F$ and $f_C$. However, it is worth noting that  
Theorem \ref{thm:theta_C<theta_F} also implies that the unfairness of $f_F$ must \emph{strictly} increase for some range of manipulation budgets, and Theorem \ref{thm:x_g<x_y} is then also sufficient  for this strict decrease of fairness to occur. Thus there always exists some range of manipulation budgets for which the fairness of the fair classifier strictly decreases in the presence of strategic agents.

\section{Multivariate Classifiers}
It is more challenging to show general results for multivariate classifiers. However, we can provide three main results that tie into both the univariate case and our subsequent empirical analyses. 
First we show existence: If the hypothesis class under consideration is at least as expressive as linear models, fairness reversal can occur under some distributions and cost functions. Second we show that when the fair classifier is, in a strict sense, more selective than the base classifier, there is always some cost function under which a fairness reversal occurs. Finally, to show that such cost functions are not necessarily esoteric, we provide an example of a natural cost function for linear classifiers that leads to fairness reversal.


\begin{theorem}
    Let $\Hp$ be a hypothesis class which is a super set of all linear models. If both $f_C, f_F$ are selected from $\Hp$, then there exists a distribution $\D$  and cost function $c$ s.t. strategic behavior leads to a fairness reversal of $f_C$ and $f_F$.
\end{theorem}

\begin{proof}
    The full proof is provided in the supplement.
\end{proof}

The next theorem states that if the set of examples which $f_F$ classifies as positive is partially a subset of those which $f_C$ classifies as positive, then there exists a cost function such that strategic agent behavior leads to $f_F$ being less fair than $f_C$.

\begin{theorem}\label{thm:subset_bad}
    Let $\X^{(f_F)}_1 = \{\x\in\X:f_F(\x) = 1\}$ and $\X^{(f_C)}_1 = \{\x\in\X:f_C(\x) = 1\}$, i.e. the set of examples predicted to be a $1$ by the respective classifier when agents are truthful. Let $U_{\D}(f, c)$ be the difference in fairness between groups on $\D$ when best responding to $f$ with cost function $c$. 
    Then there exists a cost function $c$ such that 
    \[
        U_{\D}(f_F, c) - U_{\D}(f_C, c) > U_{\D}(f_C, \infty) - \mathbb{P}_{\x\sim\D}\big(\x\in\X^{(f_F)}_1 \setminus \X^{(f_C)}_1\big)
    \]
\end{theorem}

\begin{proof}[Proof Sketch]
    The full proof is provided in the supplement. 
    Strategic agent behavior causes the decisions of a classifier to change only in one direction: negative predictions become positive. 
    Thus  $f_C|_{(\X_1^{(f_C)} \setminus \X_1^{(f_F)})}$ can be constructed by a cost function $c$ and $f_F$, yielding equal fairness between $f_C and f_F$ on $\X_1^{(f_C)}$ (giving the positive dependence of $U_{\D}(f_C, \infty)$ in the bound). 
    However, no cost function can increase the unfairness of $f_F$ on $\X_1^{(f_F)} \setminus \X_1^{(f_C)}$ (giving the negative dependence of $\Pb{\x\in \X_1^{(f_F)} \setminus \X_1^{(f_C)}}$).
\end{proof}

This result implies that the difference in unfairness between $f_F$ and $f_C$ grows as $\X^{(f_F)}_1 \setminus \X^{(f_C)}_1$ shrinks. That is, as the positive predictions of $f_F$ are more likely to be subsumed by $f_C$, the difference in unfairness grows. 
This relationship between $f_C$ and $f_F$ holds particularly well when the cost function is  \emph{congruent} to the decision boundary of $f_C$, as is the case with monotone cost functions and threshold classifiers as well as the $l_2$-norm costs and linear classifiers.

While there may be edge cases which require this cost function to be obscure for the result to hold analytically, we demonstrate in our experimental section that the $l_2$ norm is often a sufficient cost function for this result to hold.
Additionally we provide a concrete example where more selective linear classifiers exhibit this fairness reversal between $f_C$ and $f_F$ with $l_2$ manipulation cost.

\subsubsection{Linear classifiers:}
We now turn to linear classifiers as an illustrative example.
Suppose both group membership and true label have the following relationship with $\x$:  $\Pb{y=1|\x} = \varphi_y(\vv_y^T\x)$ and $\Pb{g=1|\x} = \varphi_g(\vv_g^T\x)$, where $\vv_y, \vv_g$ are fixed vectors and $\varphi_y, \varphi_g$ are a monotone functions describing a proper PDF.
To model the ``advantage'' of group 1, let $\vv_g^T\odot\vv_y > 0$ (i.e. the Hamming product is elementwise nonnegative, implying that each feature $x$ is either negatively, or positively, correlated with both $y$ and $g$.). In the supplement we show this regularly arises on real data.


\begin{theorem}\label{thm:lin_select_unfair}
    Let $f_C$ and $f_F$ be the optimal base and fair linear classifiers respectively. 
    If $f_F$ is more selective than $f_C$, i.e. $\w_C\odot\w_F > 0$ and $\theta_F > \theta_g$, then there exists a range of manipulation budgets $[B_0, B_1]$ s.t. $c(\x, \x') = ||\x - \x'||$ lead $f_F$ to be less fair than $f_C$.
\end{theorem}
 \begin{proof}
 The full proof is provided in the supplement.
\end{proof}




\section{Experiments}
\label{S:exp}

In this section we experimentally study the fairness reversal phenomenon that we so far considered theoretically.
We use three datasets commonly found in the fairness literature, 

   \noindent\textbf{Recidivism:} The COMPAS dataset in which the objective is to predict re-offending. \\
    \textbf{Law School:} Dataset of law students where the objective is to predict bar-exam passage.\\
    \textbf{Community Crime:} Dataset of communities where the objective is to predict if the community has high crime.

All three datasets have binary outcomes, and we label the more desirable outcome for the individuals by $y=1$ (e.g., \emph{not} re-offending in the recidivism data), with the less desirable outcome labeled by $y=0$.
Consequently, higher \emph{positive rate (PR)},  \emph, or \emph{false positve rate (FPR)} is more desirable for individuals.
Group membership in each dataset is determined by race, which in these datasets corresponds to a binary feature.
In all cases, we refer to the ``advantaged'' group (e.g. the group with higher \emph{PR} for \emph{PR} based fairness) as group $1$, or $G_1$, while the disadvantaged group is referred to as $0$ or $G_0$.


\begin{figure}
    \centering
    \includegraphics[scale=0.43]{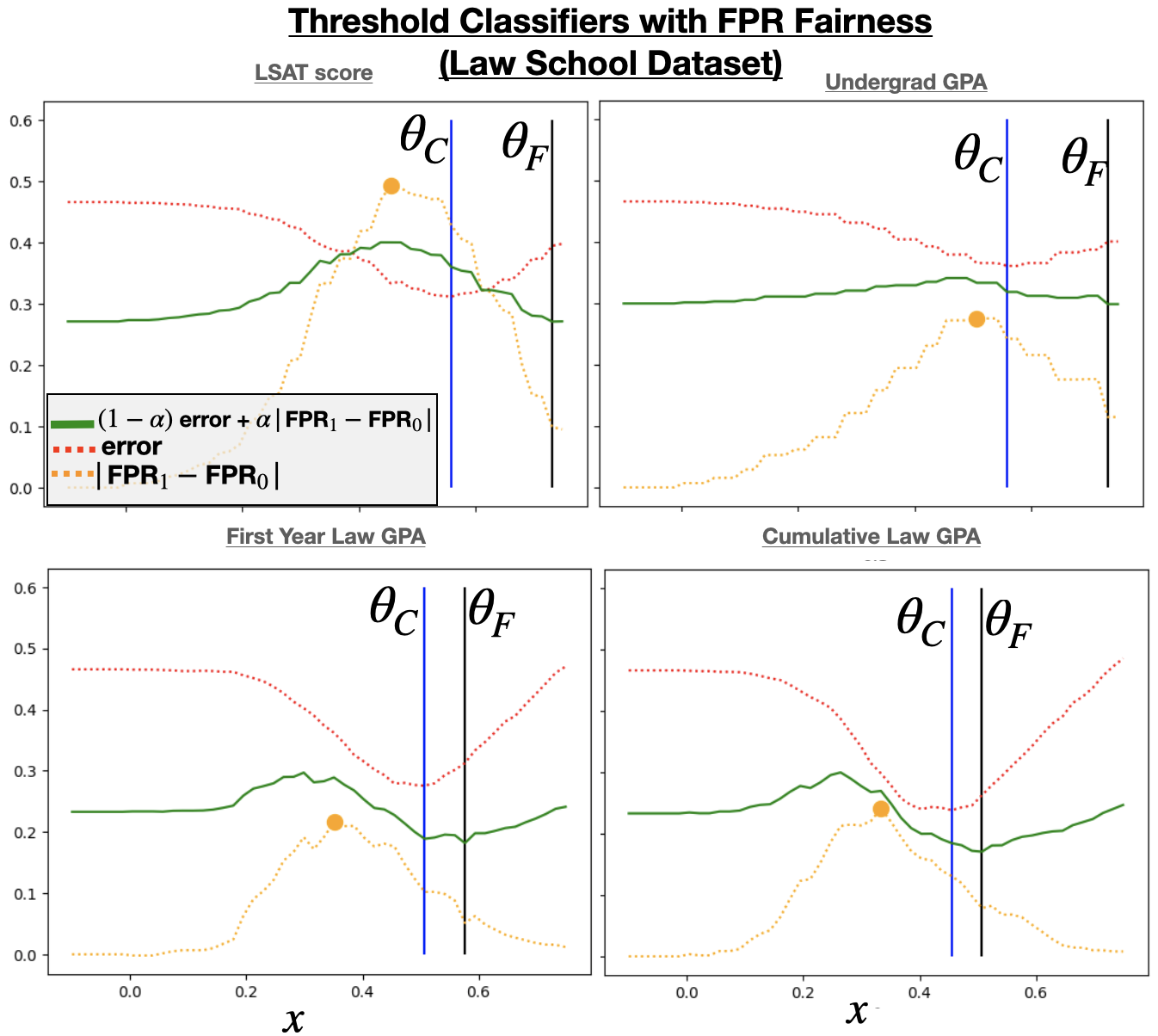}
    \includegraphics[scale=0.3]{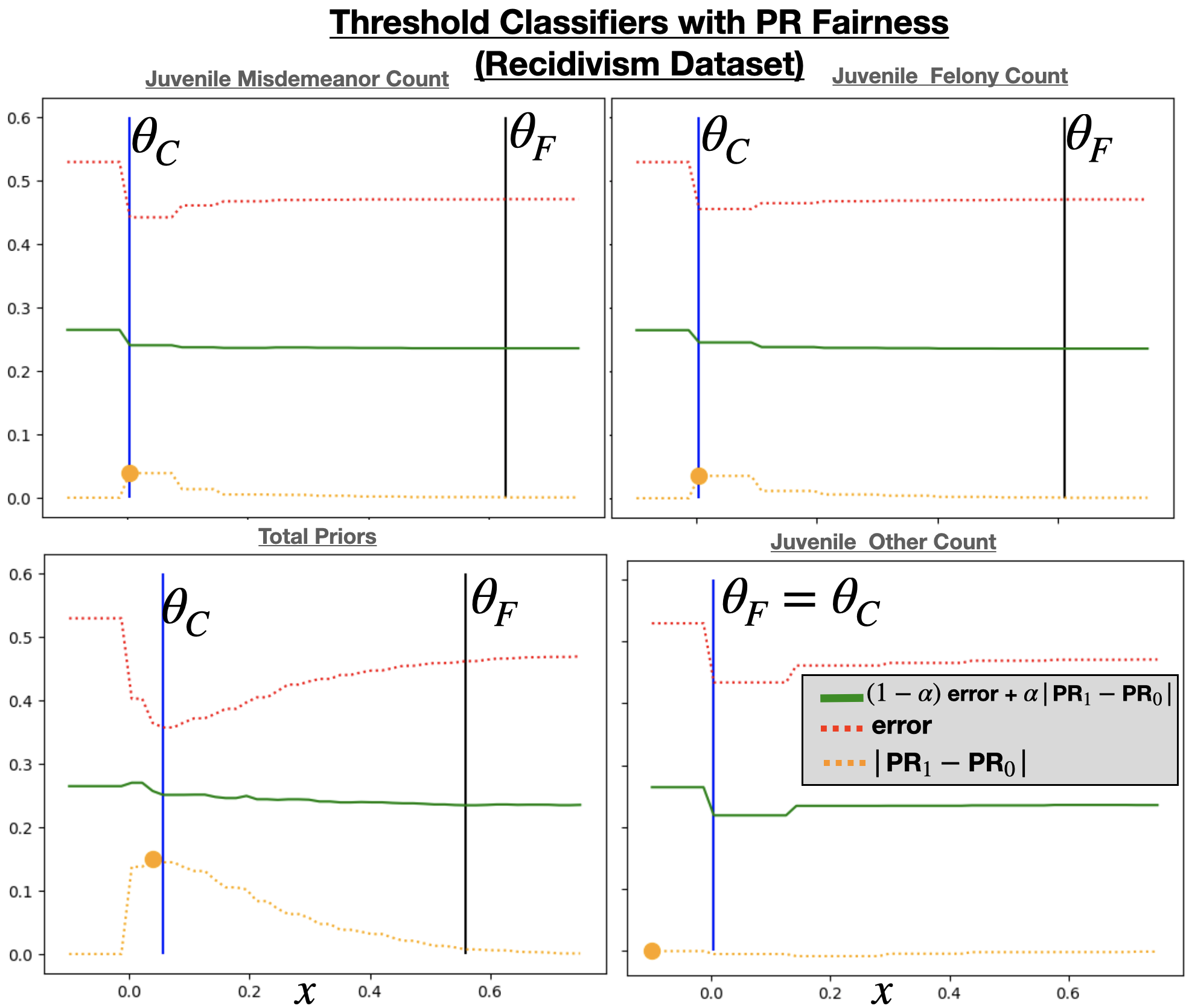}
    \caption{Optimal base and fair threshold classifiers $\theta_C, \theta_F$ respectively, on the Law School and Recidivism datasets. Recall that $\theta_F > \theta_C$ (observed in all examples, except for one - Juvenile Other Count + Recidivism) implies that strategic agent behavior will lead to $\theta_F$ becoming less fair than $\theta_C$. Error and unfairness are approximately unimodal.}
    \label{fig:single_var_FPR}
\end{figure}

We begin by considering single-variable threshold classifiers.
For each dataset we look at thresholds over ordinal features.
In both the Recidivism and Law School datasets there are 4 such features (excluding age). In the Community Crime dataset we select 4 features which have non-negligible unfairness and high correlation between the feature $x$ and label $y$. 
We then scale each feature to have domain $[0,1]$ and for visual consistency across plots, if ${\text{Cor}(x, y) < 0}$ we set $x := 1-x$.
In these three datasets most ordinal features satisfy the single crossing condition, and thus these featuers also satisfy the unimodality of error and unfairness w.r.t. to the threshold $\theta$ - this is discussed further in the supplement.

Recall that in Section \ref{sec:single_var} we showed that if error is \emph{negatively unimodal} and unfairness is \emph{positively unimodal} then $f_F$ becomes less fair than $f_C$ under strategic behavior if and only if $\theta_F \geq \theta_C$. We examine error and unfairness as a function of the selected features for each of the three datasets and each of the three fairness criteria. Due to space limitations, we present two of the figures here and defer the other seven, which are qualitatively similar, to the supplement.
Figure \ref{fig:single_var_FPR} shows that both error and unfairness are approximately unimodal across each variable for each combination of dataset and fairness metric, with the exception of one example in which unfairness is negligible for any threshold.
Moreover in these figures we see that  $\theta_C < \theta_F$ holds in all but the aforementioned exception. 

In the supplement, in addition to presenting graphs for other combinations of dataset and fairness metric, we also outline why the unimodality of error and unfairness arises so frequently. Specifically we show that when the conditional probabilities $\Pb{y=1|x}$ and $\Pb{g=1|x}$ meet the single crossing condition, as characterized in Section~\ref{sec:single_var}, unfairness, error, or both exhibit unimodality.
Moreover, this is indeed the \emph{typical} case in real data, with exceptions infrequent.

\subsection{Multivariable Classifiers}
Next, we study fairness reversal  in settings where we make use of all ordinal features in the datasets, and consider three common baseline classifiers: 
logistic regression (LGR), support vector machines (SVM with an RBF kernel), and neural networks (NN).
We consider two algorithms for group-fair classification: one due to \citet{Agarwal18} (\emph{Reductions}) and the second due to \cite{Kearns18} (\emph{GerryFair}).
Both leverage a connection to cost-sensitive learning, but the specific techniques are different.
Significant for our purposes is that both approaches turn the problem of learning with hard group-fairness constraints (that group averages on a given measure are within a specified tolerance level $\beta$) into a series of unconstrained optimization problems via an approximated Langrangian multiplier, which corresponds to $\alpha$ in our notation.
Agents' manipulations are computed via \emph{Projected Gradient Decent} (PGD)~\citep{Madry18}; this is discussed further in the supplement.

Figure \ref{fig:multi_var_unfar} shows the unfairness and error of the fair ($f_F$) and base ($f_C$) classifiers, and presents three cases in which the fairness reversal occurs and one in which it does not. These are representative cases. 
Both error (dotted) and unfairness (solid) exhibit unimodality (w.r.t. the budget $B$, which is ultimately the parameter we care about -- in the single variable case unimodality in $\theta$ implies unimodality in $B$ for any monotone cost function $c$) for both the base (blue) and fair (orange) classifiers. The shaded part indicates the region (and magnitude) of fairness reversal. In this region, the base classifier becomes more fair than than the fair classifier under strategic behavior, and there is a corresponding change where the fair classifier becomes more accurate than the base classifier, as predicted by Theorem \ref{thm:acc_reversal}.
We observe experimentally that this fairness reversal is common on both the Law School and Community Crime datasets for any combination of fairness definition, base classifier, and fair learning scheme (see Supplement).

However, while the Recidivism dataset lends itself to unfairness in the single variable case, this reversal is quite infrequent in the multivariate setting, due in part to the fact that the particularly predictive features in the other two datasets are  more directly correlated with group membership than those in the recidivism dataset (all though the correlation is in general high, it is lower than the other two datasets). This is discussed further in the supplement.

\begin{figure}
    \centering
    \includegraphics[scale=0.43]{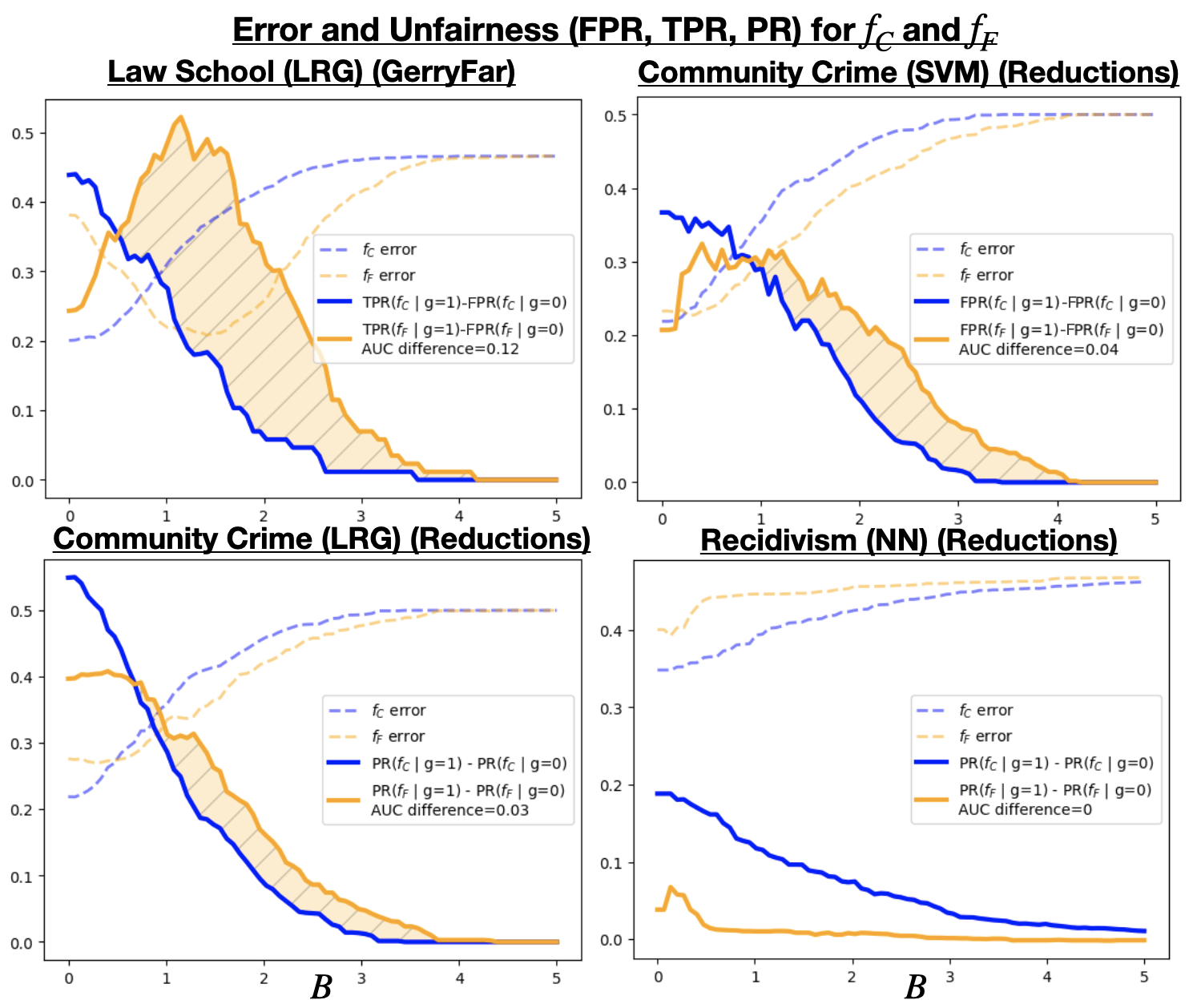}
    \caption{Error (dotted) and unfairness (solid) between $f_F$ (orange) and $f_C$ (blue) for several definitions of fairness, datasets, and learning schemes as a function of the manipulation budget $B$, with $l_2$-norm cost of manipulation.  When evaluating fairness using PR, TPR, or FPR, we observe fairness reversal for a broad range of manipulation budget $B$.}
    \label{fig:multi_var_unfar}
\end{figure}

\clearpage
\bibliographystyle{plainnat}
\bibliography{unfair}

\begin{thebibliography}{23}
\providecommand{\natexlab}[1]{#1}
\providecommand{\url}[1]{\texttt{#1}}
\expandafter\ifx\csname urlstyle\endcsname\relax
  \providecommand{\doi}[1]{doi: #1}\else
  \providecommand{\doi}{doi: \begingroup \urlstyle{rm}\Url}\fi

\bibitem[Agarwal et~al.(2018)Agarwal, Beygelzimer, Dudik, Langford, and
  Wallach]{Agarwal18}
Alekh Agarwal, Alina Beygelzimer, Miroslav Dudik, John Langford, and Hanna
  Wallach.
\newblock A reductions approach to fair classification.
\newblock In \emph{International Conference on Machine Learning}, pages 60--69,
  2018.

\bibitem[Ajunwa et~al.(2016)Ajunwa, Friedler, Scheidegger, and
  Venkatasubramanian]{Ajunwa2016HiringBA}
Ifeoma Ajunwa, Sorelle~A. Friedler, C.~Scheidegger, and S.~Venkatasubramanian.
\newblock Hiring by algorithm: Predicting and preventing disparate impact,
  2016.

\bibitem[Björkegren et~al.(2020)Björkegren, Blumenstock, and
  Knight]{Bjorkegren20}
Daniel Björkegren, Joshua~E. Blumenstock, and Samsun Knight.
\newblock Manipulation-proof machine learning.
\newblock \emph{arXiv preprint}, 2020.

\bibitem[Bolukbasi et~al.(2016)Bolukbasi, Chang, Zou, Saligrama, and
  Kalai]{bolukbasi2016man}
Tolga Bolukbasi, Kai-Wei Chang, James~Y Zou, Venkatesh Saligrama, and Adam~T
  Kalai.
\newblock Man is to computer programmer as woman is to homemaker? debiasing
  word embeddings.
\newblock \emph{Advances in neural information processing systems},
  29:\penalty0 4349--4357, 2016.

\bibitem[Buolamwini and Gebru(2018)]{Buolamwini18}
Joy Buolamwini and Timnit Gebru.
\newblock Gender shades: Intersectional accuracy disparities in commercial
  gender classification.
\newblock In \emph{Conference on Fairness, Accountability and Transparency},
  pages 77--91, 2018.

\bibitem[Chen et~al.(2020)Chen, Liu, and Podimata]{chen2020learning}
Yiling Chen, Yang Liu, and Chara Podimata.
\newblock Learning strategy-aware linear classifiers.
\newblock \emph{Advances in Neural Information Processing Systems},
  33:\penalty0 15265--15276, 2020.

\bibitem[Chouldechova(2017)]{chouldechova2017fair}
Alexandra Chouldechova.
\newblock Fair prediction with disparate impact: A study of bias in recidivism
  prediction instruments.
\newblock \emph{Big data}, 5\penalty0 (2):\penalty0 153--163, 2017.

\bibitem[Corbett-Davies and Goel(2018)]{CorbettDavies18}
Sam Corbett-Davies and Sharad Goel.
\newblock The measure and mismeasure of fairness: A critical review of fair
  machine learning, 2018.
\newblock arXiv preprint.

\bibitem[Dong et~al.(2018)Dong, Roth, Schutzman, Waggoner, and
  Wu]{dong2018strategic}
Jinshuo Dong, Aaron Roth, Zachary Schutzman, Bo~Waggoner, and Zhiwei~Steven Wu.
\newblock Strategic classification from revealed preferences.
\newblock In \emph{Proceedings of the 2018 ACM Conference on Economics and
  Computation}, pages 55--70, 2018.

\bibitem[Feldman et~al.(2015)Feldman, Friedler, Moeller, Scheidegger, and
  Venkatasubramanian]{Feldman15}
Michael Feldman, Sorelle~A. Friedler, John Moeller, Carlos Scheidegger, and
  Suresh Venkatasubramanian.
\newblock Certifying and removing disparate impact.
\newblock In \emph{ACM SIGKDD International Conference on Knowledge Discovery
  and Data Mining}, page 259–268, 2015.

\bibitem[Goodfellow et~al.(2015)Goodfellow, Shlens, and Szegedy]{Goodfellow15}
Ian~J Goodfellow, Jonathon Shlens, and Christian Szegedy.
\newblock Explaining and harnessing adversarial examples.
\newblock In \emph{International Conference on Learning Representations}, 2015.

\bibitem[Hardt et~al.(2016{\natexlab{a}})Hardt, Megiddo, Papadimitriou, and
  Wootters]{Hardt16}
Moritz Hardt, Nimrod Megiddo, Christos Papadimitriou, and Mary Wootters.
\newblock Strategic classification.
\newblock In \emph{Innovations in Theoretical Computer Science},
  2016{\natexlab{a}}.

\bibitem[Hardt et~al.(2016{\natexlab{b}})Hardt, Price, and
  Srebro]{hardt2016equality}
Moritz Hardt, Eric Price, and Nati Srebro.
\newblock Equality of opportunity in supervised learning.
\newblock \emph{Advances in neural information processing systems},
  29:\penalty0 3315--3323, 2016{\natexlab{b}}.

\bibitem[Hu et~al.(2019)Hu, Immorlica, and Vaughan]{Hu19}
Lily Hu, Nicole Immorlica, and Jennifer~Wortman Vaughan.
\newblock The disparate effects of strategic classification.
\newblock In \emph{Conference on Fairness, Accountability, and Transparency},
  2019.

\bibitem[Huang et~al.(2011)Huang, Joseph, Nelson, Rubinstein, and
  Tygar]{Huang11}
Ling Huang, Anthony~D Joseph, Blaine Nelson, Benjamin~IP Rubinstein, and J~Doug
  Tygar.
\newblock Adversarial machine learning.
\newblock In \emph{ACM Workshop on Security and Artificial Intelligence}, pages
  43--58, 2011.

\bibitem[Kearns et~al.(2018)Kearns, Neel, Roth, and Wu]{Kearns18}
Michael Kearns, Seth Neel, Aaron Roth, and Zhiwei~Steven Wu.
\newblock Preventing fairness gerrymandering: Auditing and learning for
  subgroup fairness.
\newblock In \emph{International Conference on Machine Learning}, pages
  2564--2572, 2018.

\bibitem[Lowd and Meek(2005)]{Lowd05}
Daniel Lowd and Christopher Meek.
\newblock Adversarial learning.
\newblock In \emph{ACM SIGKDD International Conference on Knowledge Discovery
  in Data Mining}, pages 641--647, 2005.

\bibitem[Madry et~al.(2018)Madry, Makelov, Schmidt, Tsipras, and
  Vladu]{Madry18}
Aleksander Madry, Aleksandar Makelov, Ludwig Schmidt, Dimitris Tsipras, and
  Adrian Vladu.
\newblock Towards deep learning models resistant to adversarial attacks.
\newblock In \emph{International Conference on Learning Representations}, 2018.

\bibitem[Milli et~al.(2019)Milli, Miller, Dragan, and Hardt]{Milli18}
Smitha Milli, John Miller, Anca~D. Dragan, and Moritz Hardt.
\newblock The social cost of strategic classification.
\newblock In \emph{Conference on Fairness, Accountability, and Transparency},
  page 230–239, 2019.

\bibitem[Vorobeychik and Kantarcioglu(2018)]{Vorobeychik18}
Yevgeniy Vorobeychik and Murat Kantarcioglu.
\newblock \emph{Adversarial machine learning}.
\newblock Morgan \& Claypool Publishers, 2018.

\bibitem[Xu et~al.(2021)Xu, Liu, Li, Jain, and Tang]{Xu21}
Han Xu, Xiaorui Liu, Yaxin Li, Anil~K Jain, and Jiliang Tang.
\newblock To be robust or to be fair: towards fairness in adversarial training.
\newblock In \emph{International Conference on Machine Learning}, 2021.

\bibitem[Zafar et~al.(2019)Zafar, Valera, Gomez-Rodriguez, and
  Gummadi]{Zafar19}
Muhammad~Bilal Zafar, Isabel Valera, Manuel Gomez-Rodriguez, and Krishna~P.
  Gummadi.
\newblock Fairness constraints: A flexible approach for fair classification.
\newblock \emph{Journal of Machine Learning Research}, 20\penalty0
  (75):\penalty0 1--42, 2019.

\bibitem[Zemel et~al.(2013)Zemel, Wu, Swersky, Pitassi, and Dwork]{Zemel13}
Rich Zemel, Yu~Wu, Kevin Swersky, Toni Pitassi, and Cynthia Dwork.
\newblock Learning fair representations.
\newblock In \emph{International Conference on Machine Learning}, pages
  325--333, 2013.

\end{thebibliography}

\clearpage
\appendix
\section*{Supplement}

\subsection*{Single crossing and unimodality}
\begin{figure}
    \centering
    \includegraphics[scale=0.25]{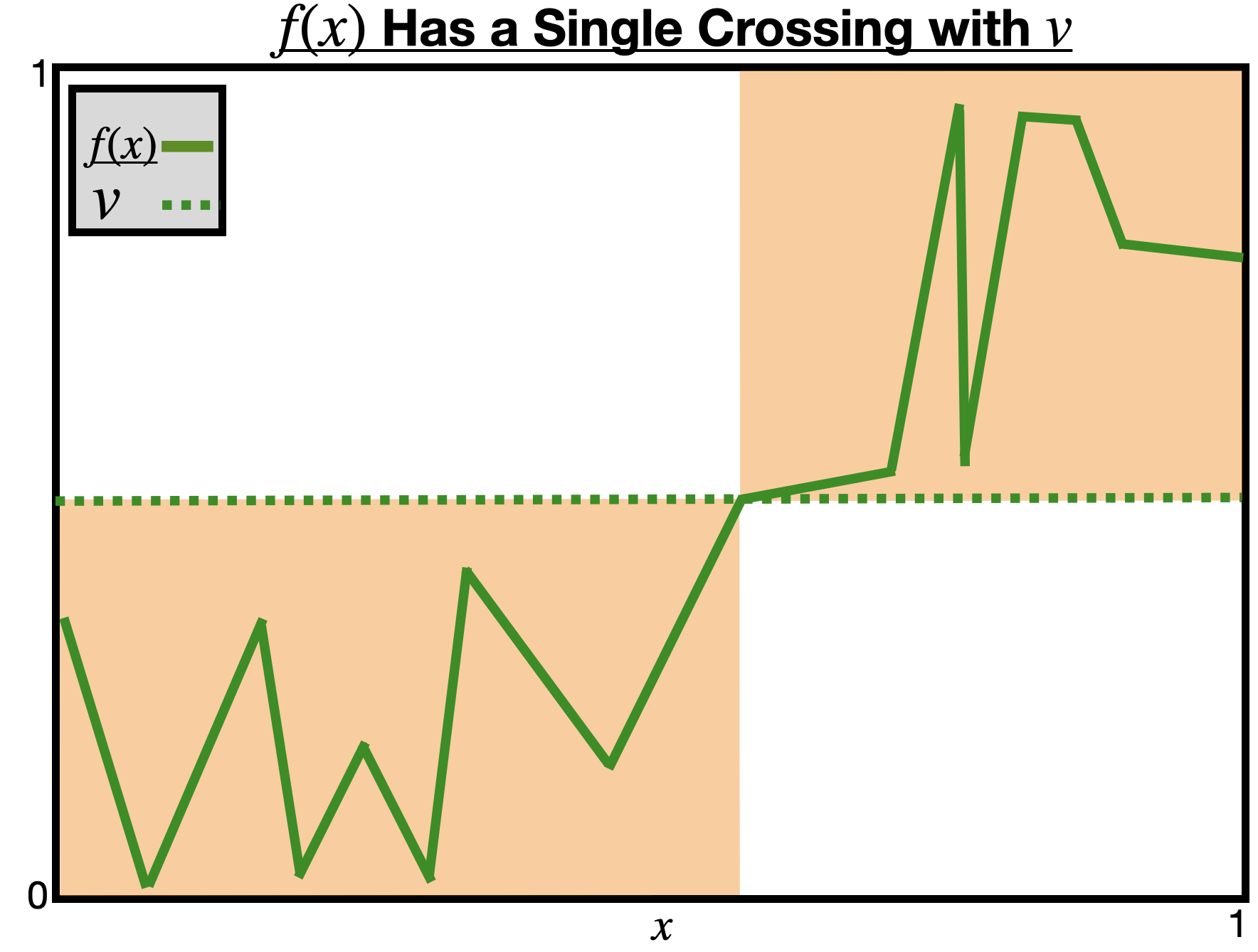}
    \caption{Example of a a function $f(x)$ (solid green) which has a \emph{single crossing} (Def \ref{def:single_cross}) with the constant function $v$ (dotted green). 
    $f(x)$ can take on any values within the orange regions and maintaining the single crossing condition with $v$. 
    That is, so long as $f(x)$ is upper bounded by $v$ prior to crossing $v$, and lower bounded by $v$ after crossing $v$, the single crossing condition holds.}
    \label{fig:my_label}
\end{figure}

\begin{lemma}\label{lem:unimodal_prop}
    A once-differentiable function is unimodal if its derivative has a single crossing with the constant function 0.
\end{lemma}
\begin{proof}
    Let $f(x):\mathbb{R} \rightarrow \mathbb{R}$ be a  once-differentiable function. Suppose that $f'(x)$ has a single crossing with $0$. Then there exists a point $z$ such that $x \leq z$ implies $f'(x) \leq 0$ and $z \leq x$ implies $0 \leq f'(x)$. 
    Thus $f$ is monotonically decreasing on the interval $(-\infty, z]$ and monotonically increasing on the interval $[z, \infty)$. Implying that $f$ is unimodal.
\end{proof}

\subsection*{Proofs}
\begin{proof}[Proof: (Lemma \ref{lem:error_convex})]
    The error of a classifier $f(x) = \mathbb{I}[x \geq \theta]$ is given by,
    \begin{align*}
          & 1 - \Pb{\mathbb{I}[x \geq \theta] = y}\\
        =~&1 - \Pb{x \geq \theta, y=1} - \Pb{x \leq \theta, y=0}\\
        =~&\Pb{y=0} + \Pb{x \leq \theta, y=1} - \Pb{x \leq \theta, y=0}
    \end{align*}
    Since $x$ is a continuous random variable and the terms involving $\theta$ are joint CDFs with well defined conditional PDFs, the derivative of the above expression w.r.t. to $\theta$, exists and is equal to
    \begin{align*}
        &h_{y, x}(y = 1, x=\theta) - h_{y, x}(y = 0, x=\theta)\\
        =~&h_x(x=\theta)\big(\mathbb{P}(y = 1|x=\theta) - \mathbb{P}(y = 0| x=\theta)\big)\\
        =~&h_x(x=\theta)\big(2\mathbb{P}(y = 1| x=\theta) - 1\big)
    \end{align*}
    Since $\Pb{y=1|x=\theta}$ is \emph{split} by the value $\nicefrac{1}{2}$, the above derivative is \emph{split} by the value $0$, thus by Lemma \ref{lem:unimodal_prop} error is \emph{negatively unimodal} with global minima at any $\theta_C$ s.t. ${\Pb{y=1|x=\theta_C} = \nicefrac{1}{2}}$.
\end{proof}

\begin{proof}[Proof: (Lemma \ref{lem:PR_concave})]
    For a classifier $f{(x)} = {\mathbb{I}[x \geq \theta]}$, we begin by demonstrating (1) the \emph{unimodality} of ${\Pb{x \geq \theta|g=1} - \Pb{x \geq \theta|g=0}}$ and then use this propriety to show (2) the equivalence between  ${\Pb{x \geq \theta|g=1} - \Pb{x \geq \theta|g=0}}$ and the unfairness term ${\big|\Pb{x \geq \theta|g=1} - \Pb{x \geq \theta|g=0}\big|}$.
    First, note that
    \begin{align*}
        =~&\Pb{x \geq \theta|g=1} - \Pb{x \geq \theta|g=1}\\
        =~&\frac{\Pb{g=1, x\geq \theta}}{\Pb{g=1}} - \frac{\Pb{g=0, x\geq \theta}}{\Pb{g=0}}\\
        =~&\frac{\Pb{g=1} - \Pb{g=1, x\leq \theta}}{\Pb{g=1}} - \frac{\Pb{g=0} - \Pb{g=0, x\leq \theta}}{\Pb{g=0}}\\
        =~&1 - \frac{\Pb{g=1, x\leq \theta}}{\Pb{g=1}} - 1 + \frac{\Pb{g=0}  \Pb{g=0, x\leq \theta}}{\Pb{g=0}}\\
        =~&-\frac{\Pb{g=1, x\leq \theta}}{{\Pb{g=1}}} + \frac{\Pb{g=0, x\leq \theta}}{{\Pb{g=0}}}
    \end{align*}
   
    
    Since each term involving $\theta$ is a joint CDF, the derivative of this term w.r.t to $\theta$ exists and is equal to
    \begin{align*}
          &\frac{h_{g, x}(g=0, x=\theta)}{{\Pb{g=0}}} - \frac{h_{g, x}(g=1, x=\theta)}{{\Pb{g=1}}}\\
        =~&\frac{\Pb{g=0| x=\theta}h_{x}(x=\theta)}{{\Pb{g=0}}} - \frac{\Pb{g=1| x=\theta}h_{x}(x=\theta)}{{\Pb{g=1}}}\\
        =~&\frac{\big(1 - \Pb{g=1| x=\theta}\big)h_{x}(x=\theta)}{{\Pb{g=0}}}\\
        &\quad- \frac{\Pb{g=1| x=\theta}h_{x}(x=\theta)}{{\Pb{g=1}}}\\
        =~&h_{x}(x=\theta)\frac{\Pb{g=1} - \Pb{g=1|x=\theta}}{\Pb{g=1}\Pb{g=0}}
    \end{align*}
    Since $\Pb{g=1|x}$ is \emph{split} by the value $\Pb{g=1}$ the above term is \emph{split} by the value 0, thus by Lemma \label{lem: unimodal_prop} the term  $\Pb{x \geq \theta|g=1} - \Pb{x \geq \theta|g=0}$ is \emph{positively unimodal}, and is maximized at any $\theta_U$ s.t. 
    \begin{align*}
        &h_{x}(x=\theta_U)\frac{\Pb{g=1} - \Pb{g=1|x=\theta_U}}{\Pb{g=1}\Pb{g=0}} = 0
    \end{align*}
    Since $h_x(x=\theta) > 0$ any such $\theta_U$ has the propriety that  $\Pb{g=1|x=\theta_U} = \Pb{g=1}$.
    Thus concluding the proof of (2). 
    
    We now use (2) to show that (1) immediately follows.
    Note that for $\theta \in \{0, 1\}$ we have $\Pb{x \geq \theta|g=1} = \Pb{x \geq \theta|g=0}$. Since the function is \emph{positively unimodal} and $\Pb{g=1} > 0$ neither $\theta =0 $ nor $\theta=1$ can be points corresponding to local maximums, hence for any $\theta$ we have
    \begin{align*}
    &~\Pb{x \geq \theta|g=1} - \Pb{x \geq \theta|g=0}\\
    \geq&~\Pb{x \geq 1|g=1} - \Pb{x \geq 1|g=0}\\
    =&~0
    \end{align*}
\end{proof}

\begin{proof}[Proof: (Lemma \ref{lem:TPR_FPR_concave}]
    This proof follows a similar argument to the the proof of \ref{lem:PR_concave}. 
\end{proof}

\begin{proof}[Proof: (Lemma \ref{lem:thresh_same_as_strat_behave})]
    When all agents prefer positive predictions to negative predictions, their manipulations will change the classifier in only a single direction, namely manipulations cause negatively predicted examples to become positively predicted. 
    Thus, only agents with feature $x$, where $f(x) = 0$ need be considered.
    
    Suppose $f$ is a threshold classifier with threshold $\theta$, then the agent's best response to $f$ is,
    \begin{align*}
        x^* = \text{argmax}_{x}&\mathbb{I}[x' \geq \theta] - \mathbb{I}[x \geq \theta]\\
        \text{s.t. }&~c(x, x') \leq B
    \end{align*}
    Since the cost function $c(x, x')$ is monotone w.r.t. $|x'-x|$ the above best response has solution
    \begin{align*}
        x^* = 
         \begin{cases}
            \theta &\text{ if } c(x, \theta) \leq B~\text{ and }x < \theta\\
             x & \text{ otherwise}
        \end{cases}
    \end{align*}
    Moreover, the monotonicity of $c(x, x')$ also implies that if an agent with feature $x$ has best response $x^*=\theta$, then so will any other agent with $x_1$ where $x \leq x_1 < \theta$.
    
    Thus, the distribution shift of $\D$ caused by strategic behavior, can be quantified in terms of the agent with the  smallest feature which is able to report a value of $\theta$, i.e. the feature 
    \begin{align*}
        x_{\text{min}} = \text{argmin}_{x}& x\\
        \text{s.t. }& c(x, \theta) \leq B
    \end{align*}
    Thus when agents are strategic, any agent with feature $x \geq  x_{\text{min}}$ will be positively classified by $f$. 
    Therefore, the threshold $\theta'=x_{\text{min}}$ makes the same classifications on the unmanipulated distribution $\D$ as $\theta$ makes on the manipulated distribution $\D_{\theta}^{(B)}$.
\end{proof}

\begin{proof}[Proof: (Theorem \ref{thm:theta_C<theta_F})]
    We first show that $\theta_C < \theta_F$ implies the existence of a budget interval $[B_1, B_2]$ s.t. strategic agent behavior under any $B \in [B_1, B_2]$ leads to $f_F$ being less fair than $f_C$. We then show that if $\theta_F < \theta_C$, no such budget interval exists.
    
    The unfairness of threshold $\theta$ w.r.t. to the distribution $\D$ and fairness metric $\M\in\{\PR, \TPR,\FPR\}$ is expressed as,
    \begin{align*}
        U(\theta, \D) = \big| \M_{\D}(\theta|g=1) - \M_{\D}(\theta|g=0)\big|,
    \end{align*}
    For a given threshold $\theta$ and manipulation budget $B$ the best response of an agent with true type $a = (g, x)$ is
    \begin{align*}
        x_{\theta}^{(B)} = \text{argmax}_{x'}\big(\mathbb{I}[x' \geq \theta] - \mathbb{I}[x \geq \theta]\big)~\text{ s.t. } c(x, x') \leq B,
    \end{align*}
    When agents, originally distributed in accordance with $\D$, play this optimal responses w.r.t. $\theta$ and $B$, let the resulting distribution be $\D^{(B)}_{\theta}$.
    The difference in unfairness, between classifiers, when agents are strategic, can then be expressed as $U(\theta_C, \D^{(B)}_{\theta_C}) - U(\theta_F, \D^{(B)}_{\theta_F})$.
    Lemma \ref{lem:thresh_same_as_strat_behave} gives a way to express this difference in terms of the original distribution $\D$ by changing the thresholds, namely
    \begin{align*}
        &~U(\theta_C, \D^{(B)}_{\theta_C}) - U(\theta_F, \D^{(B)}_{\theta_F})\\
        =&~U(\theta_C^{(B)}, \D) - U(\theta_F^{(B)}, \D)
    \end{align*}
    where 
    \begin{align*}
        &\theta_C^{(B)} = \text{argmin}_{x}x \text{ s.t. } c(x, \theta_C) \leq B \text{ and,}\\
        &\theta_F^{(B)} = \text{argmin}_{x}x \text{ s.t. } c(x, \theta_F) \leq B
    \end{align*}
    By the monotonicity of $c(x, x')$ w.r.t. $x'-x$ we have that 
    \begin{align*}
        &\theta_C < \theta_F \implies \theta_C^{(B)} \leq \theta_F^{(B)}\quad\forall~B\geq 0 \\
        &\theta_C > \theta_F \implies \theta_C^{(B)} \geq \theta_F^{(B)}\quad\forall~B\geq 0
    \end{align*}
    i.e. the relative ordering of the thresholds is preserved under strategic behavior for any manipulation budget $B$, this fact will be of use later.
    
    Let $\theta_U = \text{argmax}_{\theta}U(\theta, \D)$, we now proceed to prove the forward direction of our claim by three cases of the relative order of the thresholds $\theta_C, \theta_F, \theta_U$:
    \begin{align*}
        &1.)~\theta_C < \theta_F \leq \theta_U\\
        &2.)~\theta_C \leq \theta_U \leq \theta_F\\
        &3.)~\theta_U \leq \theta_C < \theta_F
    \end{align*}
    
    First note that case (1) is infeasible. By Lemmas \ref{lem:PR_concave} and \ref{lem:TPR_FPR_concave}, we know that $U(\theta, \D)$ is \emph{positively unimodal} and maximized at $\theta_U$. Therefore, for $\theta \in [\theta_C, \theta_U]$  we have that $U(\theta, \D)$ is monotonically increasing. 
    Thus in case (1) we have ${U(\theta_F, \D) \geq U(\theta_C, \D)}$, which is impossible since $f_F$ is assumed to be strictly more fair than $f_C$.
    
    To prove that the claim holds in cases (2) and (3) we use the fact that the unfairness term $U(\theta, \D)$, being \emph{positively unimodal} implies that the term is also monotonically increasing on the interval $[0, \theta_U]$ and monotonically decreasing on $[\theta_U, 1]$. Hence, as stated previously, for any $\theta_1\leq \theta_2 \leq \theta_U$ it must also be the case that $U(\theta_1, \D) \leq U(\theta_2, \D) \leq U(\theta_U, \D)$. 
    Therefore it suffices to show that there exists a budget interval $[B_1, B_2]$ s.t. for any $B\in[B_1, B_2]$ we have $\theta_C^{(B)} \leq \theta_F^{(B)} \theta_U$, and as stated previously, for any $B\geq 0$ $\theta_C < \theta_F$ implies that $\theta_C^{(B)} \leq \theta_F^{(B)}$. Hence we need only show that in cases (2), (3) having $B\in[B_1, B_2]$ implies $\theta_F^{(B)} \leq \theta_U$.
    
    In both cases, (2) and (3), this follows immediately form Lemma \ref{lem:thresh_same_as_strat_behave}, which gives the existence of a budget $B_U$, such that $\theta_F^{(B_U)} = \theta_U$, and implies that $\theta_F{(B)}$ is monotonically decreasing w.r.t. to $B$. Therefore, for any $B\in[B_U, \infty)$ it must be the case that $U(\theta_C^{(B)}, \D) \leq U(\theta_F^{(B)}, \D)$.
    
    To prove the reverse direction, we need to show that when $\theta_F < \theta_C$ it is the case that for any $B\geq 0$ we have $U(\theta_F^{(B)},\D) \leq U(\theta_C^{(B)}, \D)$.
    We again show this by three cases on the relative order of the thresholds $\theta_F, \theta_C, \theta_U$:
    \begin{align*}
        &1.)~\theta_F < \theta_C \leq \theta_U\\
        &2.)~\theta_F \leq \theta_U \leq \theta_C\\
        &3.)~\theta_U \leq \theta_F < \theta_C
    \end{align*}
    Similar to the forward direction of the proof, one case is infeasible, namely case (3). This can be seen be a symmetric argument to the previous one, specifically that on the interval $[\theta_U, \theta_C]$ both error and unfairness are monotonically decreasing, and thus $\theta_F$ could not be an optimal fair threshold.
    
    As shown previously, when $\theta_F < \theta_C$ it is also the case that for any $B \geq 0$ we have $\theta_F^{(B)} \leq \theta_C^{(B)}$, and if $\theta_C^{(B)} \geq \theta_U$ then $U(\theta_F^{(B)}, \D) \leq U(\theta_C^{(B)}, \D)$. Thus the claim holds for case (1), leaving only case (2) left to prove.
    
    In case (2) we have $\theta_F \leq \theta_U \leq \theta_C$. Let $B_U$ the budget s.t. $\theta_C{(B_U}) = \theta_U$, then for $B\in [0, B_U]$ the term $U(\theta_C^{(B)}, \D)$ is monotone increasing, while $U(\theta_F^{(B)}, \D)$ is monotone decreasing, and thus $U(\theta_F^{(B)}, \D) \leq U(\theta_C^{(B)}, \D) $. 
    Moreover for $B \in [B_U, \infty)$ we have already have show that $U(\theta_F^{(B)}, \D) \leq U(\theta_C^{(B)}, \D)$.
    
    Therefore the reverse direction of the claim holds, and thus there exists an interval $[B_1, B_2]$ s.t. ${U(\theta_C^{(B)}, \D) \geq U(\theta_F^{(B)}, \D)}$ for $B\in [B_1, B_2]$ iff $\theta_C < \theta_F$. 
\end{proof}

\begin{proof}[Proof: (Theorem \ref{thm:x_g<x_y})]
    Given $\alpha \in (0, 1)$, fairness metric ${\M\in\{\PR, \TPR, \FPR\}}$, and data distribution $\D$, the objective of learning scheme is to find $\theta_F$ such that
    \begin{equation}\label{eq:fair_obj}
        \theta_F = \text{argmin}_{\theta}(1-\alpha)\Pb{\mathbb{I}[x \geq \theta] \neq y} + \alpha U_{\D}(\theta)
    \end{equation}
    where
    \begin{align*}
         U_{\D}(\theta) = \big|\M(\theta|g=1) - \M(\theta|g=0)\big|
    \end{align*}
    By Lemma \ref{lem:error_convex} the error term $\Pb{\mathbb{I}[x \leq \theta] = y}$ is negatively unimodal unimodal in $\theta$ and achieves a minimum at $\theta_C$ where $\Pb{y=1|x=\theta_C} = \mathbb{P}(y=1)$.
    Similarly, by Lemmas \ref{lem:PR_concave}, \ref{lem:TPR_FPR_concave} the unfairness term $U_{\D}(\theta)$ is positively unimodal in $\theta$ and achieves a maximum at $\theta_U$ where $\Pb{g=1|x=\theta_U}=\Pb{g=1}$.
    Thus for any $\alpha$ the fair learning objective (Equation \ref{eq:fair_obj}) is monotonically increasing, thus $\theta_F \notin [\theta_U, \theta_C]$.
    Implying that either $\theta_F \in[0, \theta_U)$ or $\theta_F \in (\theta_C, 1]$.
    Moreover since 
\end{proof}

\begin{proof}[Proof: (Theorem \ref{thm:acc_reversal})]
     Since the data is linearly separable, and the definition of $\theta_C$, we have $\theta_C = \theta^*$ which has an accuracy of 1. Due to the manipulation from $\X(\theta_C,B)$, the accuracy of $\theta_C$ on $\D'$ is then precisely $1-\Pb{x \in \X(\theta_C, B)}$.
    
    Define the following two quantities:
    \begin{align*}
    x_C := \inf_{x \in \X(\theta_C, B)} x,~~~x_F := \inf_{x \in \X(\theta_F, B)} x.
    \end{align*}
    Consider the following two cases. First when $x_F \leq \theta_C$. Since all $x \geq \theta_C$ would either be classified by $\theta_F$ as 1 or be included in $\X(\theta_F, B)$ (so will report a $x'$ which will be classified as 1), the accuracy of $\theta_F$ is precisely $1-\Pb{x \in [\x_F, \theta_C]}$. Since $c$ is monotonic in $|x-x'|$, and $\theta_F > \theta_C$, we have $x_F > x_C$, and therefore 
    \[
    1-\Pb{x \in [x_F, \theta_C]} > 1-\Pb{x \in [x_C, \theta_C]}
    \]
    establishing $\theta_F$ is more accurate. 
    
    Now consider the case $x_F > \theta_C$. In this case, the error of $\theta_F$ is incurred by the $x \in [\theta_C, x_F]$ that has a $y=1$ but has no incentive to deviate into $\X(\theta_F,B)$ and be misclassified. Therefore the accuracy of $\theta_F$ is:
    \begin{align*}
       &1-\Pb{x \in [\theta_C, x_F]}\\
       =&1-\bigl(\Pb{x \in [\theta_C, \theta_F]} -\Pb{x \in \X(\theta_F, B)}  \bigr)\\
       \geq & 1-\Pb{x \in \X(\theta_C, B)},
    \end{align*}
    where the inequality is due to the condition we required.
\end{proof}

\begin{proof}
    This proof follows a similar line of reasoning to its single dimensional counter part, namely that agents best responding to linear classifiers $f_C$ and $f_F$ can equivalently formulated as those same agents truthfully reporting to two modified linear classifiers $f_C$ and $f_F$. More specifically, agents sourced from $\D$ best responding to a classifier, say $f$ with cost function $c(\x, \x') = ||\x-\x'||$ and budget $B$, causes a shift in $D$. We call this shift $\D_{f}^{(B)}$, when agents strategically react to $f$, the classifier $f$ is now predicting over $\D_{f}^{(B)}$, rather than $\D$. However, there exists a modified linear classifier $f'$ whose predictions on $\D$ are precisely those of $f$ on $\D_{f}^{(B)}$. Thus our proof strategy is to prove that there exists a $B_0$ s.t. for any $B \geq B_0$, the corresponding modified base and fair classifier $f_C'$ and $f_F'$ have the propriety that $f_F'$ is less fair than $f_C'$ on $\D$.
    
    To construct $f_C'$ and $f_F'$ we first compute the best response of agents. Under the cost function of $c(\x, \x') = ||\x - \x'||$ and a classifier linear classifier with parameters $\w, \theta$ the optimal response of an agent with true feature $\x$ is to play 
      \begin{align*}
        \x' = 
        \begin{cases}
        \x - \frac{\w^T\x + \theta}{||\w||_2^2}\w &\text{ if } \w^T\x \leq \theta \text{ and } \big|\big| \frac{\w^T\x + \theta}{||\w||_2^2}\w\big|\big|\leq B \\
        \x & \text{ otherwise}
        \end{cases}
    \end{align*}
    Using this best response and the fact that $\w_C, \w_F$ are unit vectors we can assume for the sake of analysis, that every agent plays $\x' = \x - B\w$ since doing so will not effect the decisions of $f$.
    Therefore we can write $f_C'$ and $f_F'$ as $\w_C, \theta_F - B$ and $\w_F, \theta_F-B$ respectively.
    
    Thus, we need only show that for some $B$, thresholds $\theta_F-B$ and $\theta_C-B$ have ``reversed" fairness compared to $\theta_F$ and $\theta_C$. 
    Since $\w_C \odot \w_F >0$, it must also be the case that $\w_F \odot \vv_g >0$ and $\w_C \odot \vv_g> 0$, which in turn implies that $\mathbb{P}(g=1|\x)$ is also monotone in $\w_F^T\x$ and $\w_C^T\x$. 
    As a result, if there are no agents which $f_F'$ predicts positively, but $f_C'$ predicts negativity, then $f_C'$ is at least as fair as $f_F'$.
    More specifically, let
    \begin{align*}
        S_{0,B} = \{\x\in\X: \w_F\x < \theta_F - B \text{ and } \w_C\x\geq\theta_F - B\},\\
        S_{1,B} = \{\x\in\X: \w_F\x \geq \theta_F - B \text{ and } \w_C\x<\theta_F - B\}.
    \end{align*}
    The fairness of $f_C'$ and $f_F'$ are equal on $\X\setminus\big(S_{0, B} \cup S_{1, B})$, since on this set their decisions agree. 
    We show that if $S_{1,B'} = \emptyset$ for some $B'$, then there exists a budget $B_0 \geq B'$  s.t. for all $B \geq B_0$. The existence of $B'$ comes directly from the fact that $\theta_F > \theta_C$. Such a $B'$ must exist since $\theta_C < \theta_F$.
    
    Intuitively this follows a similar argument to the single variable case, namely that $S_{1, B'} = \emptyset$ implies that the hyperplane induced by $f_F'$ is an upperbound for the hyperplane induced by $f_C'$. Meaning that unfairness is again unimodal w.r.t. the budget $B$.
    
    More specifically, the PR fairness of a linear classifier can be expressed as 
    \begin{align*}
          &\Pb{f(x) =1 | g=1} - \Pb{f(x) =1|g=0}\\
         =&\Pb{\w^T\x \geq \theta | g=1} - \Pb{\w^T\x\geq \theta | g=0}
    \end{align*}
    Thus if $S_{1, B} = \emptyset$ then the difference in unfairness of $f_F'$ and $f_C'$ can be expressed as
    \begin{align*}
         =&\Pb{\w_F^T\x \geq \theta_F - B | g=1, \x\in S_{0, B}}\Pb{\x\in S_{0, B}} \\
         &\quad- \Pb{\w_F^T\x\geq \theta_F-B | g=0, \x\in S_{0, B}}\Pb{\x\in S_{0, B}}  \\
         &\quad-\Pb{\w_C^T\x \geq \theta_C - B | g=1, \x\in S_{0, B}}\Pb{\x\in S_{0, B}}  \\
         &\quad+ \Pb{\w_C^T\x\geq \theta_C-B | g=0, \x\in S_{0, B}}\Pb{\x\in S_{0, B}}\\
         =&\Pb{\w_C^T\x\geq \theta_C-B, \w_F^T\x< \theta_F-B | g=0}\\
         &\quad-\Pb{\w_C^T\x \geq \theta_C - B, \w_F^T\x< \theta_F-B | g=1}\\
         =&\bigg(\frac{\Pb{g=0|\w_C^T\x \geq \theta_C - B, \w_F^T\x< \theta_F-B} }{\Pb{g=0}}\\
         &\quad-\frac{\Pb{g=1|\w_C^T\x \geq \theta_C - B, \w_F^T\x< F - B} }{\Pb{g=1}}\bigg)\Pb{\x\in S_{0, B}}
    \end{align*}
    Since we only care that this quantity is non-negative, i.e. $f_C'$ is at least as fair as $f_F'$, the term $\Pb{\x\in S_{0, B}}$ can be dropped.
    The difference in conditionals can be written as 
    \begin{align*}
         &\int_{\x\in S_{0, B}} \frac{1 - \varphi_g(\vv_g^T\x)}{\Pb{g=0}}d\x - \int_{\x\in S_{0, B}} \frac{\varphi_g(\vv_g^T\x)}{\Pb{g=1}}d\x\\
        =&\int_{\x\in S_{0, B}} \frac{\Pb{g=1} - \varphi_g(\vv_g^T\x)}{\Pb{g=0}\Pb{g=1}}
    \end{align*}
    Since $\varphi_g$ is monotone in both $\w_C^T\x$ and $\w_F^T\x$ and decreasing $B$ can only introduce values of $\x$ into $S_{0, B}$ which have smaller values of $\w_C^T\x$ and remove values which have higher value of $\w_F^T\x$. 
   Thus for $B$ such that any $\x\in S_{0, B}$ has $\varphi_g(\vv_g^T\x) \leq \mathbb{P}(g=1)$ the integral is guaranteed to be positive, implying that for $B$ or any larger budget $f_C'$ is more fair than $f_F'$.
   Again due to the monotonicity of $\varphi_g$, such a sufficiently large $B$ must exist.
\end{proof}

\subsection*{Experiments}
\subsubsection*{Single Crossing}
Figures \ref{fig:single_cross_1}, \ref{fig:single_cross_2}, \ref{fig:single_cross_3} show the single crossing conditions between $\Pb{y=1|x}$, and $\Pb{g=1|x}$, and their respective constant functions given in Lemmas \ref{lem:error_convex}, \ref{lem:PR_concave}, \ref{lem:TPR_FPR_concave}. 
We see that in all three datasets the single crossing conditions approximately holds in the sense that when the condition is violated, (i.e. crossing the respective horizontal line more than once) the violation is small in magnitude. Recall that the single crossing propriety implies the unimodality of the error and unfairness terms. Small violations (both in magnitude and duration) of the single crossing condition amount to small changes in the derivative of error or unfairness, which in term does not consequentially impact the unimodality of either term from an empirical perspective.

With this said we see that the Recidivism dataset breaks the single crossing assumptions on $\Pb{g=1|x}$ more so than the other two datasets.
However, we still observe both the unimodality of unfairness as well as a reversal of fairness between the base and fair classifier on this dataset (in the single variable case).

Moreover, we see that the variables in the Recidivism dataset have weaker relationships between $\mathbb{P}(y=1|x)$ and $\mathbb{P}(g=1|x)$ compared to those of the Law School and Community Crime datasets. Meaning that in the Recidivism dataset, label identification and group identification are easier to untangle, compared to the other two datasets.
This is partially the reason we observe that the fairness reversal does not occur on the recidivism dataset in the multivariate setting.

\begin{figure}
    \centering
    \includegraphics[scale=0.58]{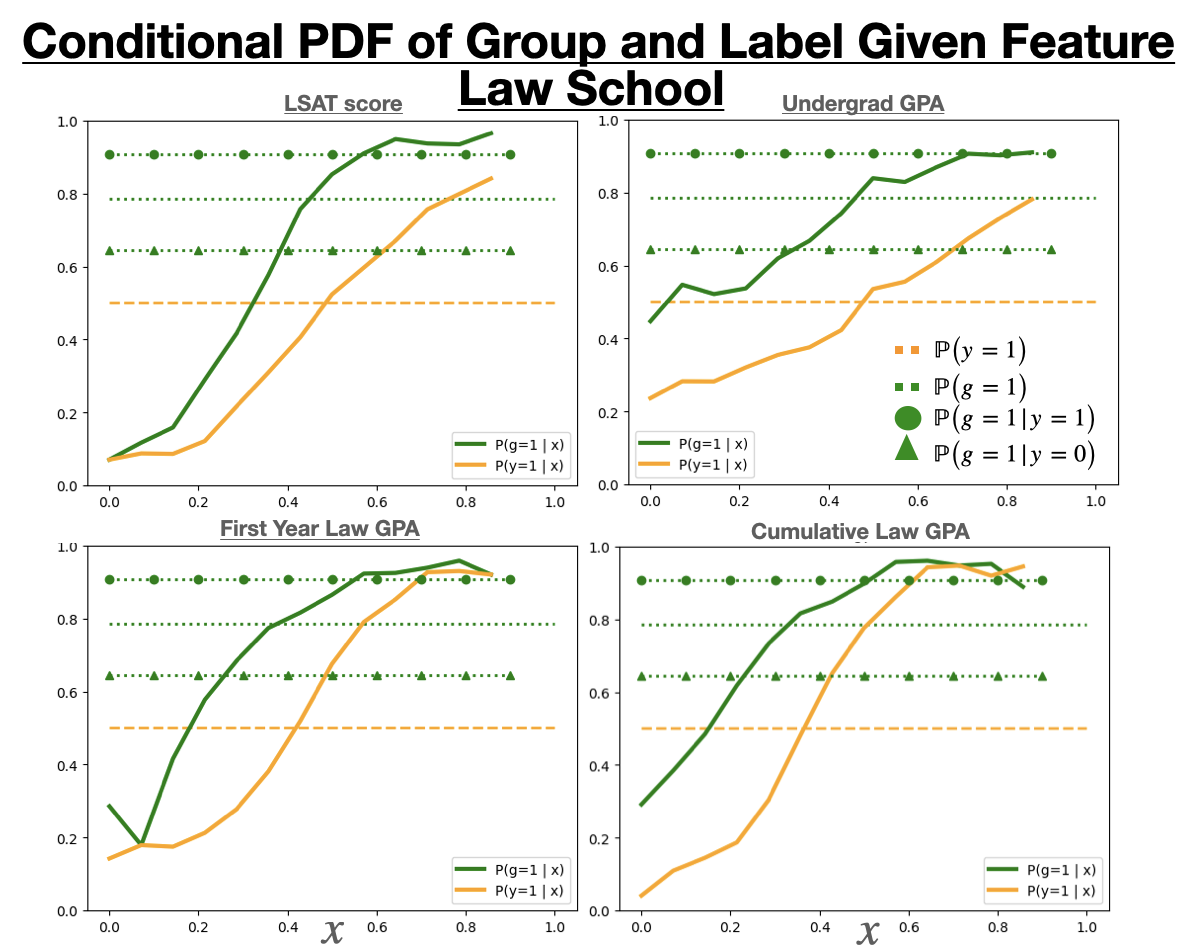}
    \caption{Probabilities of group membership $g$ (green) and true label $y$ (orange). Probabilities conditioned on the feature $x$ are given as solid lines, while those unconditioned are given as dotted, or dashed, lines.
    Recall that if the conditioned probabilities $\mathbb{P}(g=1|x)$ and $\mathbb{P}(y=1|x)$ having a single crossing with the respective unconditioned value (outlined in Lemmas \ref{lem:error_convex}, \ref{lem:PR_concave}, \ref{lem:TPR_FPR_concave}) then error and unfairness will be unimodal w.r.t. to the threshold $\theta$. For example, in the case of \PR fairness, if $\mathbb{P}(g=1|x)$ has a single crossing with $\mathbb{P}(g=1)$ and $\mathbb{P}(y=1|x)$ has a single crossing with $\mathbb{P}(y=1)$ then error and unfairness are unimodal w.r.t. to $\theta$.} 
    \label{fig:single_cross_1}
\end{figure}
\begin{figure}
    \centering
    \includegraphics[scale=0.6]{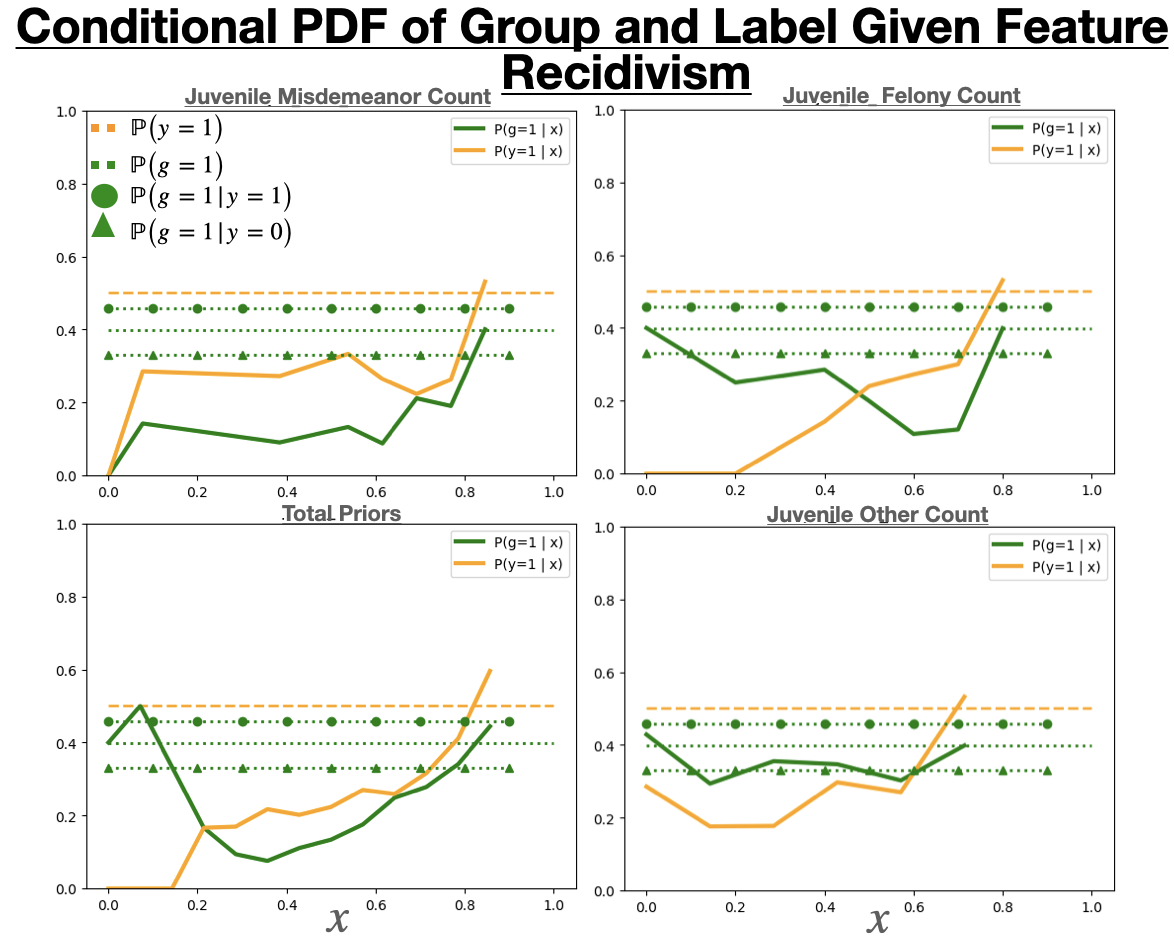}
    \caption{}
    \label{fig:single_cross_2}
\end{figure}
\begin{figure}
    \centering
    \includegraphics[scale=0.58]{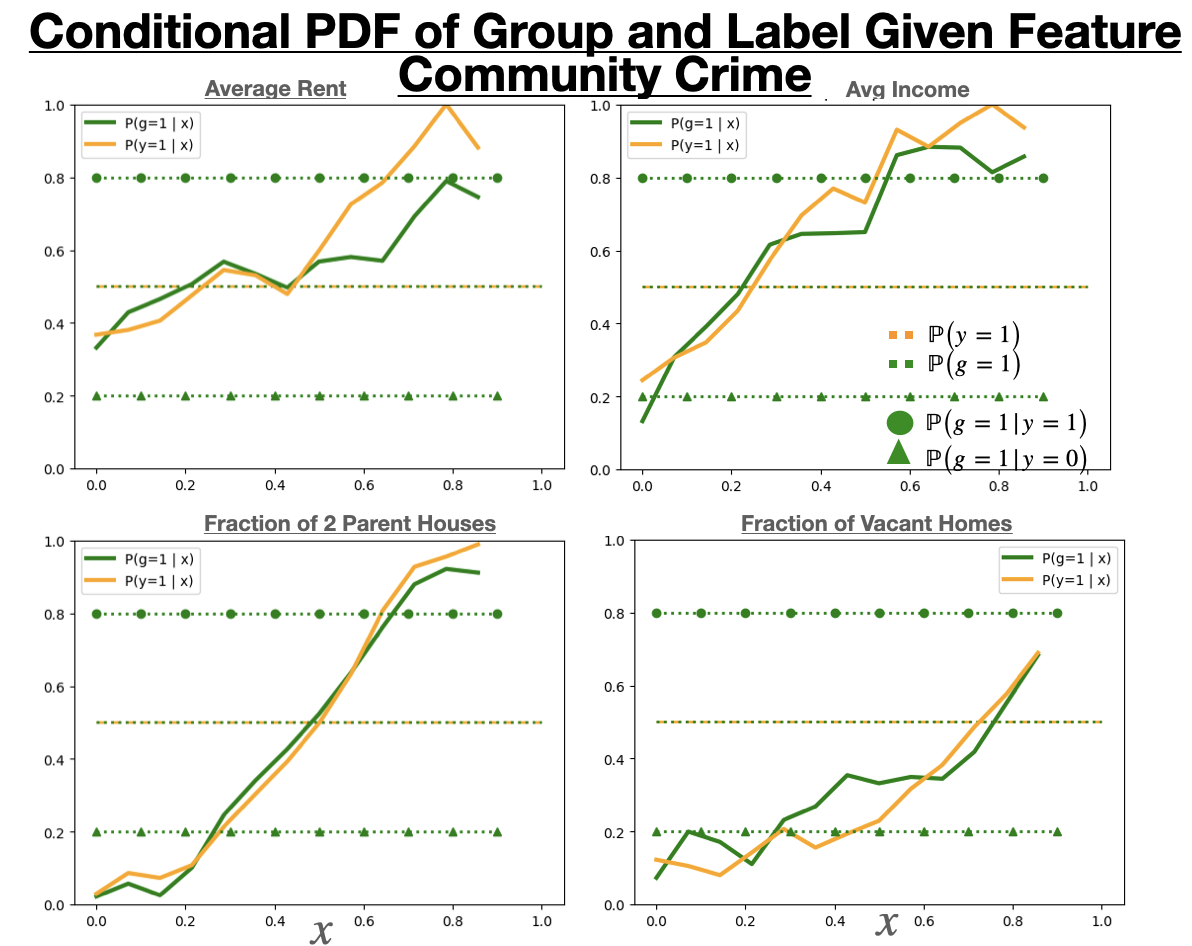}
    \caption{}
    \label{fig:single_cross_3}
\end{figure}

\subsubsection{Fair learning schemes}
We make use of two fair learning algorithms to generate the fair models (denoted as $f_F$), namely GerryFair and Reductions.
Each algorithm takes as input a base-learner (not to be confused with the baseline classifier which we denote as $f_C$). 
This base-learner is used solve the fair learning objective through cost sensitive learning. 
Each algorithm uses their respective base learner in a unique way, and the fair models produced by each learning scheme different considerably in terms of their structure.
Reductions uses the base-learner to perform traditional cost sensitive learning and outputs a model which is of the same type of the base-learner. For example if the base-learner is Logistic Regression, then $f_F$ is also a Logistic Regression model. 
Thus the Projected Gradient Decent attack (PDG) is effective at computing an agents best response to $f_F$ when $f_F$ is learned via reductions and a differentiable base-learner (e.g. Logistic Regression, SVM, and Neural Networks).
In the case of GerryFair the returned fair model $f_F$ has a different structure from the base learner, namely $f_F$ is an ensemble of models produced from the base learner. Thus the resulting model may not be smooth and PGD will not work to compute agents best response. However, of the models we examine (LRG, SVM, NN) GerryFair only supports LRG and SVM with a linear kernel. Meaning that each learner in the ensemble, produced by GerryFair, is linear and thus it is trivial  to compute each agent's best response.

\subsubsection*{Fairness Reversal}
Recall that in the single variable case, strategic manipulation leads to a fairness reversal between the base and fair thresholds $\theta_C$ and $\theta_F$ respectively, if and only if $\theta_C < \theta_F$. Figures \ref{fig:single_cross_1}-\ref{fig:single_var_9} show the relationship between $\theta_C$ and $\theta_F$ for each of the variables, dataset, and fairness metrics we study. In these figures we see that $\theta_C < \theta_F$ is a common. Moreover, we see that the cases where this relationship does not hold are cases in which either $\theta_C < \theta_U$ (meaning the sufficient condition of Theorem \ref{thm:x_g<x_y} does not hold), the fair classifier is trivial (i.e. $\theta_F=0$), or there is negligible unfairness regardless of the value selected for $\theta$.
Moreover, we see that both error and unfairness are unimodal w.r.t. $\theta$, thus Lemma \ref{lem:thresh_same_as_strat_behave} implies that error and unfairness will remain unimodal w.r.t. the manipulation budget $B$ for \emph{any} manipulation cost function $c(x, x')$ which is monotone in $|x' - x|$.

With respect to Figures \ref{fig:single_cross_1}-\ref{fig:single_var_9}, agent manipulation amounts to ``moving" each threshold to the left. We can see that when $\theta_C < \theta_F$, moving $\theta_C$ to the left decreases unfairness, while moving $\theta_F$ to left increases unfairness, until the manipulated $\theta_F$ has been moved all the way to $\theta_U$ (the most unfair threshold). 
Additionally in these figures we see that not only does this leftward shift  increase the unfairness of $\theta_F$, but also increased the accuracy of $\theta_F$: a phenomenon outlined by Theorem \ref{thm:acc_reversal}. That is, in the cases where $\theta_C < \theta_F$, strategic manipulation leads to both a fairness, and an accuracy, reversal between $\theta_C$ and $\theta_F$.

In the multivariate case, Figures \ref{fig:CC_LRG_red}-\ref{fig:multi_2}, show that again the fairness reversal is common. In these figures the error (dotted) and unfairness (solid) are given for $f_C$ (blue) and $f_F$ (orange) as a function of the manipulation budget $B$ when agents have manipulation cost $c(\x, \x') = ||\x-\x'||$.
The shaded orange region indicates both the duration (in terms of $B$) and magnitude of the fairness reversal of $f_F$ and $f_C$.
Similar to the single variable case, these graphs again show a fundamental trade-off between fairness and accuracy in the presence of manipulation.
In all settings where a fairness reversal between $f_F$ and $f_C$ occurs, an accuracy reversal also occurs. 
Namely, if $f_F$ becomes less fair than $f_C$, it also becomes more accurate than $f_C$. 
Moreover, as was the case in the single variable case, we see that in the multivariate case both error and unfairness exhibit unimodality w.r.t. to the budget $B$.

In the single variable case, we would expect that once $f_C$ and $f_F$ respectively hit the point with maximum unfairness (as a function of $B$) their unfairness would decrease at an equal rate from that point onward since both classifiers are effectively sharing the same unfairness curve, but sit at different points. 
In the multivariate case, we make this same observation. After reaching the most unfair $B$, both classifiers decreases at similar rates. However, $f_C$ requires a larger $B$, than $f_C$, to reach this point. Which ultimately leads to $f_F$ becoming less fair, since the unfairness of $f_F$ is still increasing while the unfairness of $f_C$ has already begun to fall.

The fairness reversal is common on the Law School and Community Crime datasets, but non-existent on the recidivism dataset. As mentioned previously there is a weaker correlation between group membership and true label on the Recidivism dataset compared to the other two datasets. We suspect that this weaker link between accuracy and unfairness results in fair classifiers which lose less of their fairness in the presence of strategic behavior, since the agents benefiting from strategic behavior should be more representative of the population as a whole, rather than one particular group.

We train a total of 9 linear models in our experiments (3 fairness definitions across 3 datasets using LRG and Reductions). We observe that 5/9 cases the propriety that  $\w_C \odot \w_F > 0$ approximately holds. Here approximately holding means that if an element pair ${w_{i, C}w_{i, F} < 0}$ the respective magnitude of the product is small relative to those of the larger value weights.

\begin{figure}
    \centering
    \includegraphics[scale=0.52]{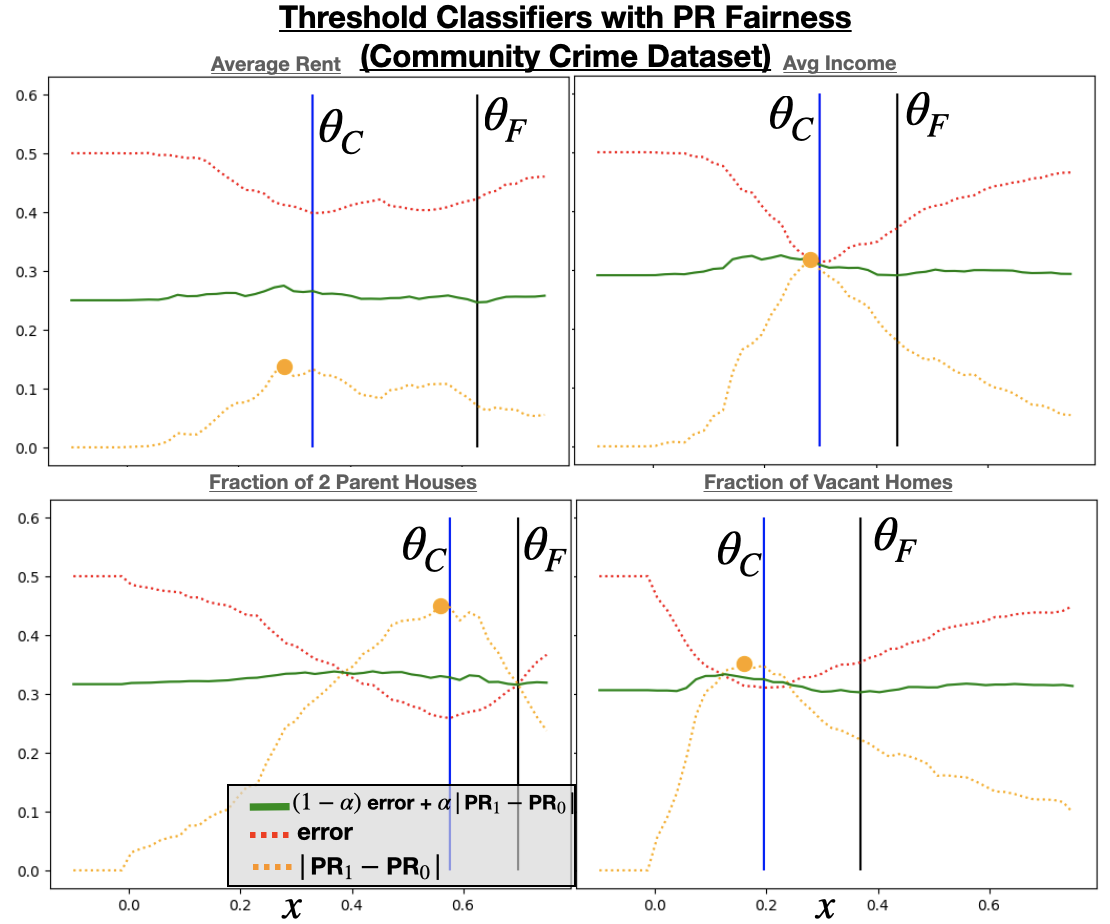}
    \caption{Unfairness and error of threshold classifiers. Both error and unfairness are approximately unimodal w.r.t. threshold $\theta=x$. 
    Thus error and unfairness are also unimodal w.r.t. the manipulation budget $B$ for any manipulation cost function $c(x, x')$ which is monotone in $|x'-x|$. When this unimodality holds $\theta_C < \theta_F$ implies that strategic manipulation will lead to $\theta_C$ becoming more fair than $\theta_F$. This fairness reversal is due to the fact that strategic manipulation amounts to lowering (shifting to the left) the threshold. In this figure, as well as the subsequent figures, we see that $\theta_C < \theta_F$ is a common occurrence (namely 30 our of the 36 combinations of variable, fairness metric, and dataset studied).}
    \label{fig:single_var_1}
\end{figure}
\begin{figure}
    \centering
    \includegraphics[scale=0.52]{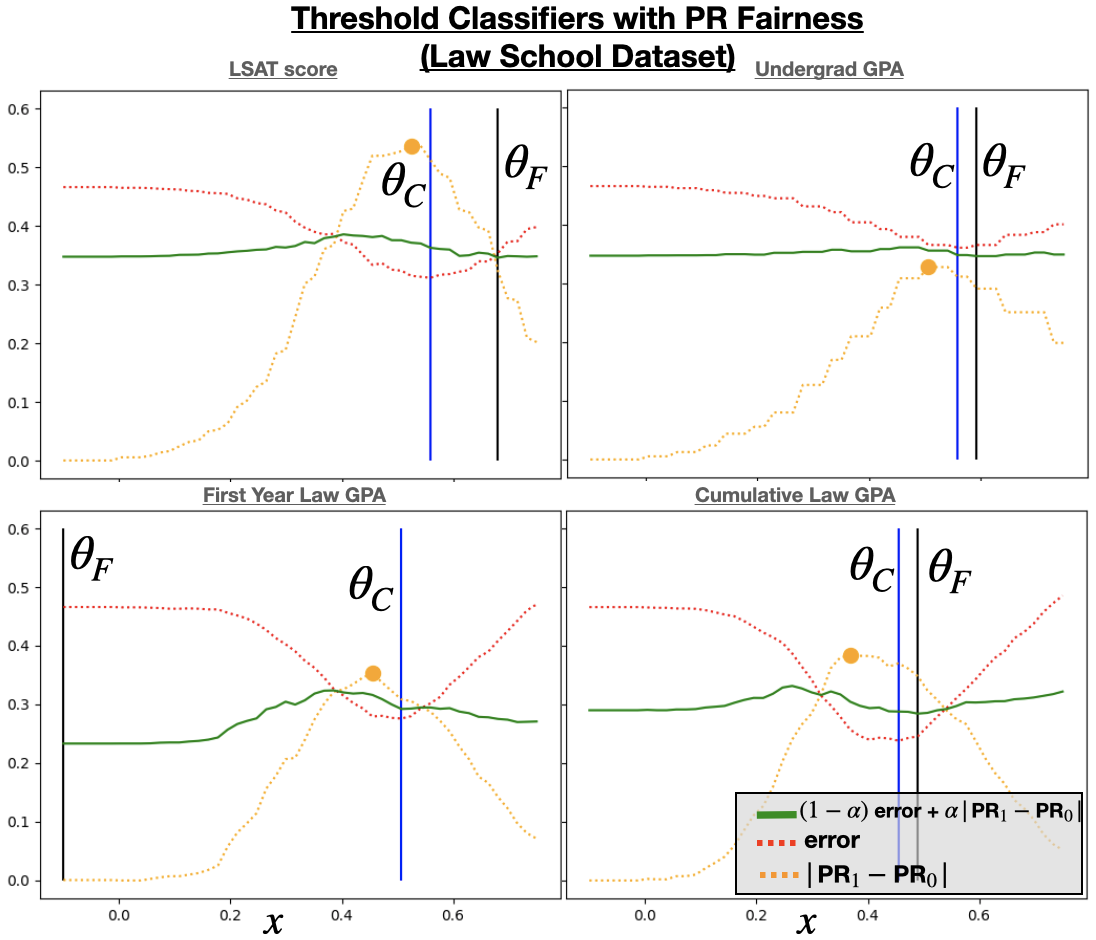}
    \caption{}
    \label{fig:single_var_2}
\end{figure}
\begin{figure}
    \centering
    \includegraphics[scale=0.52]{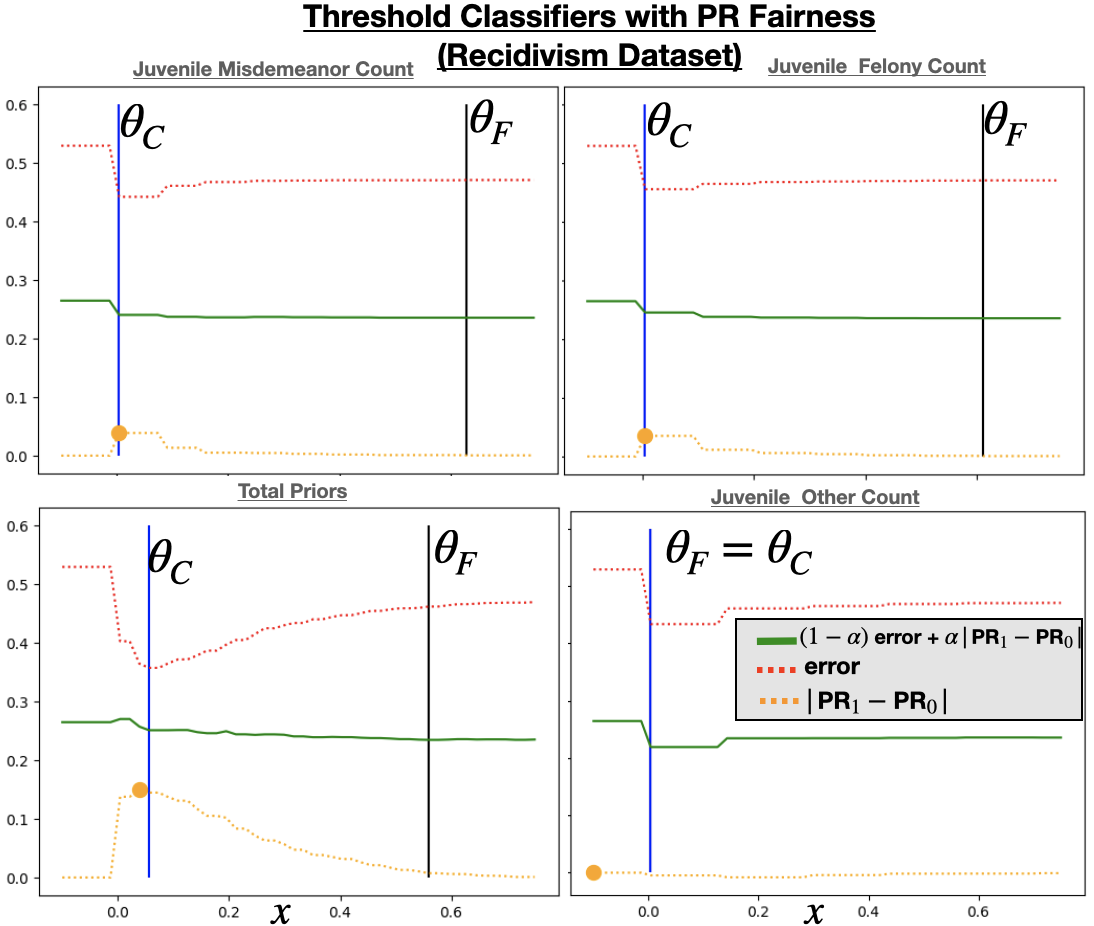}
    \caption{}
    \label{fig:single_var_3}
\end{figure}
\begin{figure}
    \centering
    \includegraphics[scale=0.52]{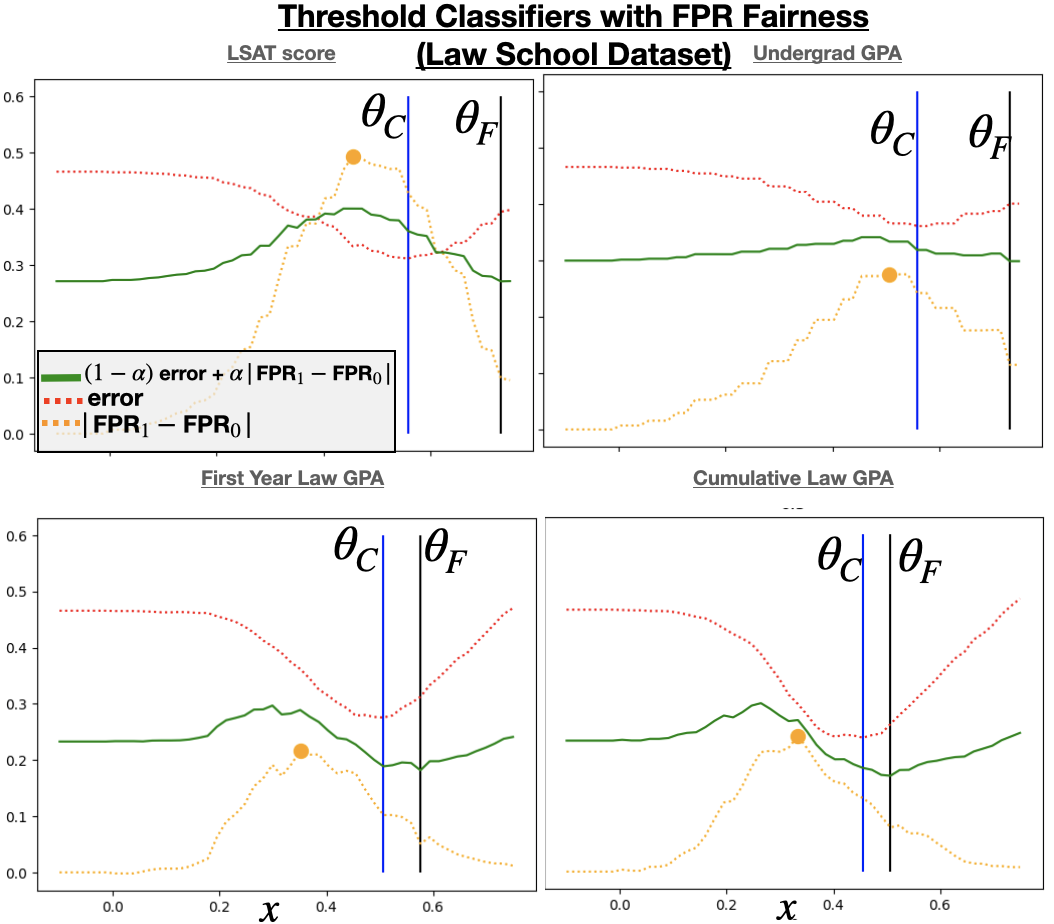}
    \caption{}
    \label{fig:single_var_4}
\end{figure}
\begin{figure}
    \centering
    \includegraphics[scale=0.52]{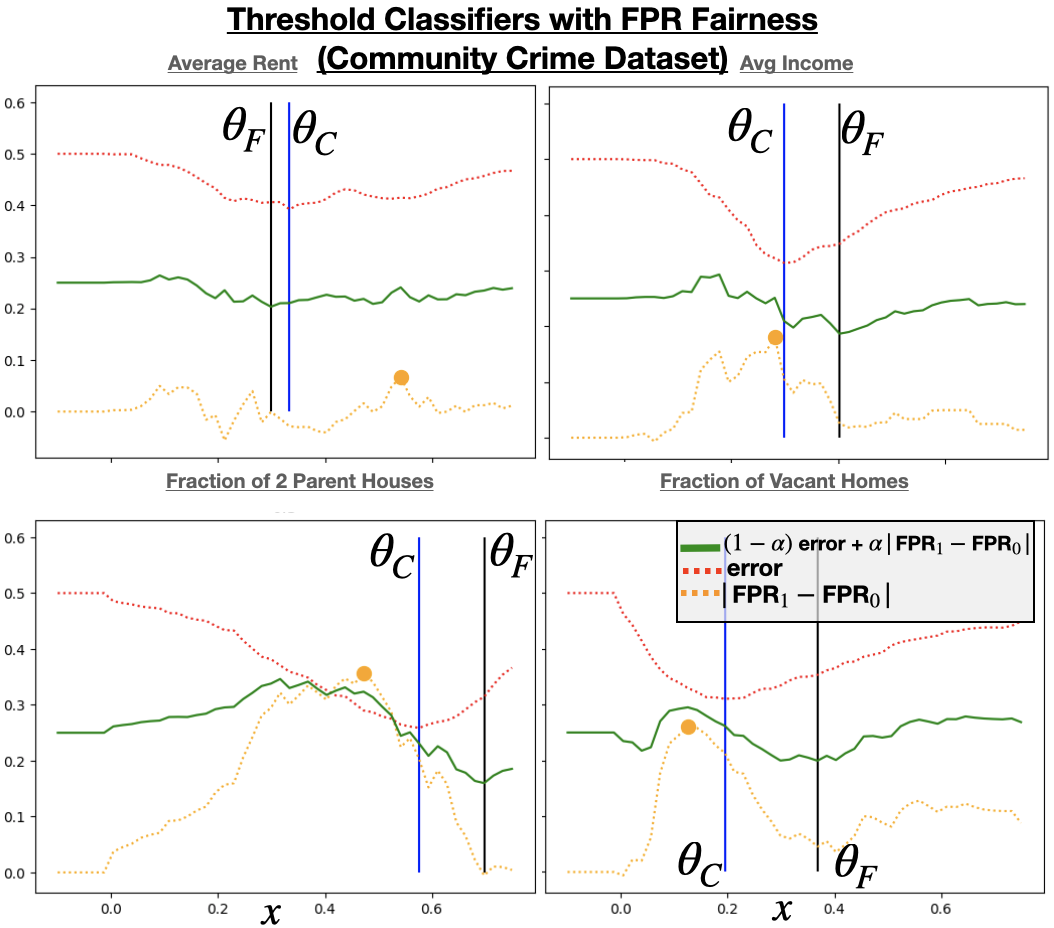}
    \caption{}
    \label{fig:single_var_5}
\end{figure}
\begin{figure}
    \centering
    \includegraphics[scale=0.52]{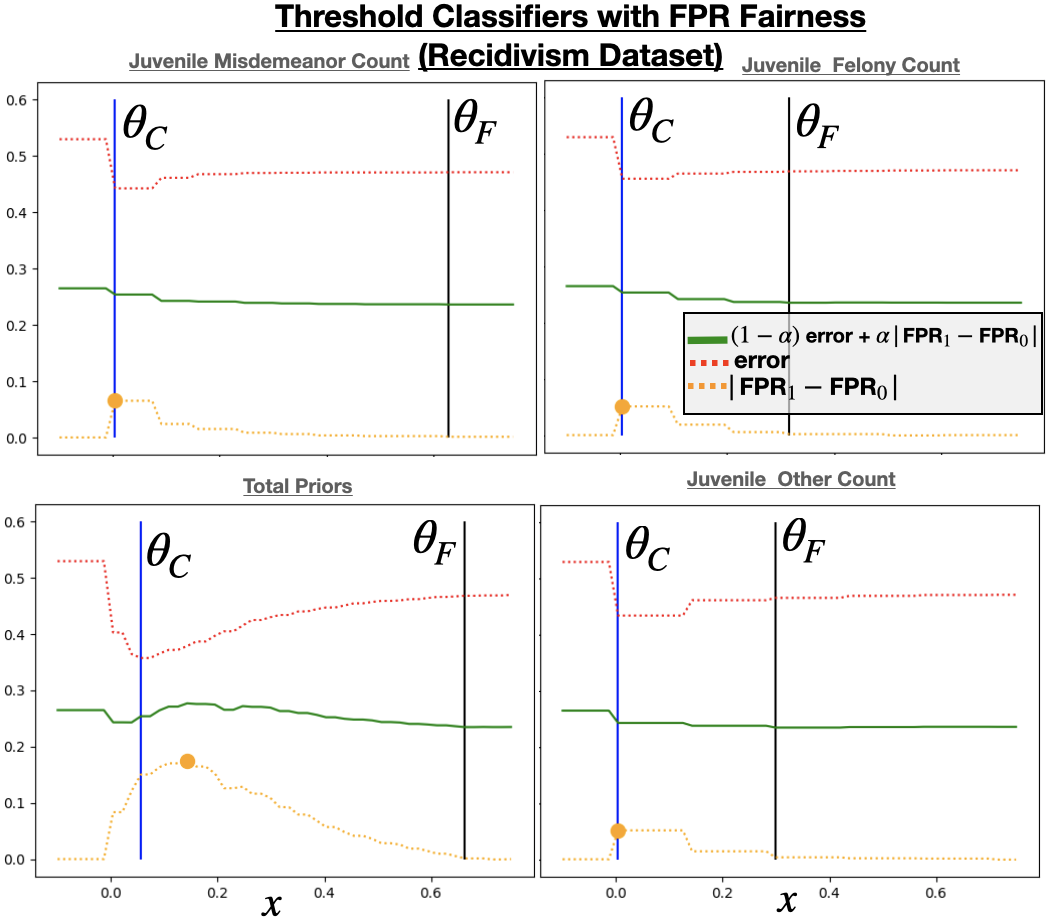}
    \caption{}
    \label{fig:single_var_6}
\end{figure}
\begin{figure}
    \centering
    \includegraphics[scale=0.52]{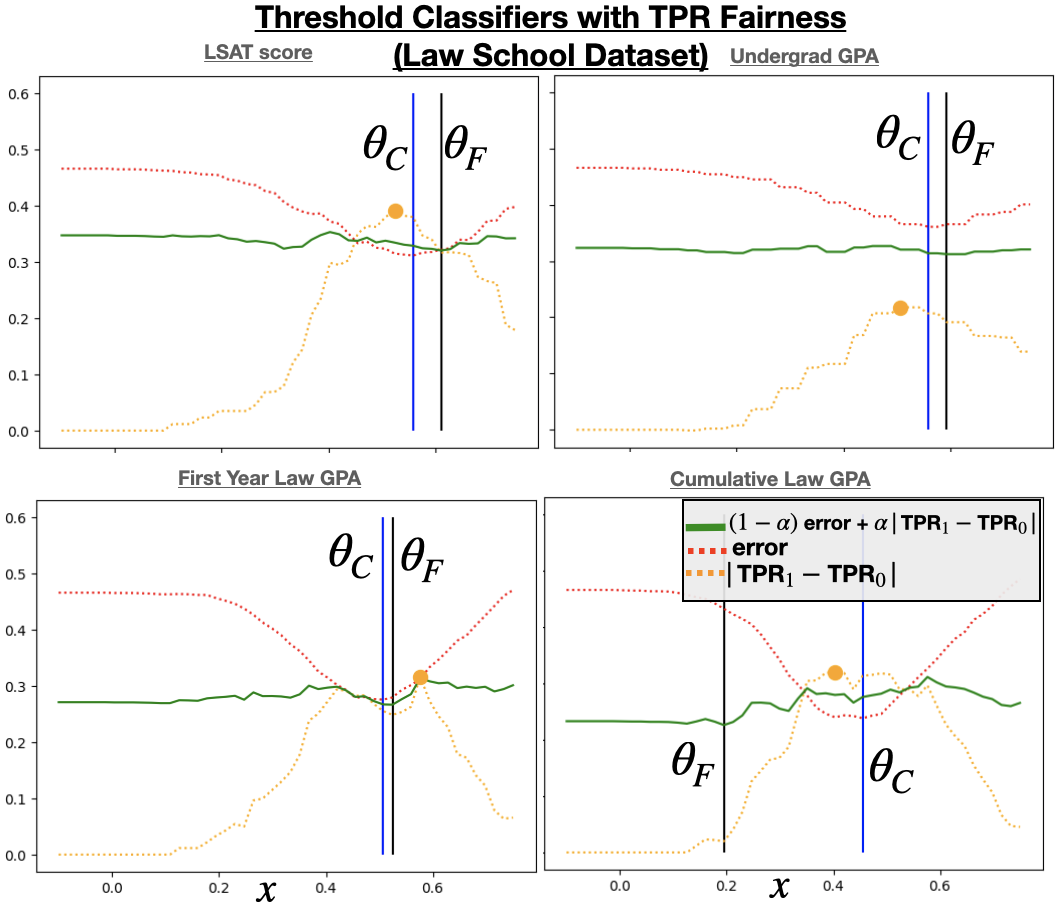}
    \caption{}
    \label{fig:single_var_7}
\end{figure}
\begin{figure}
    \centering
    \includegraphics[scale=0.52]{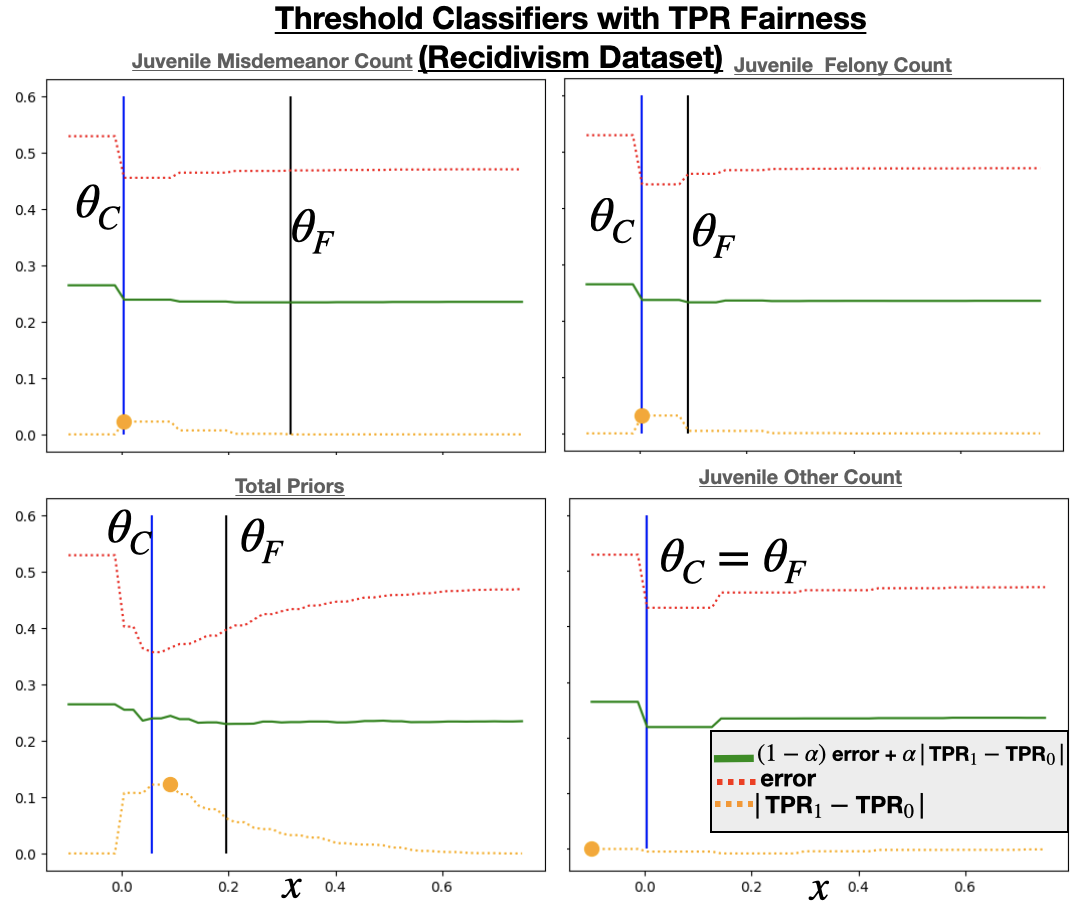}
    \caption{}
    \label{fig:single_var_8}
\end{figure}
\begin{figure}
    \centering
    \includegraphics[scale=0.52]{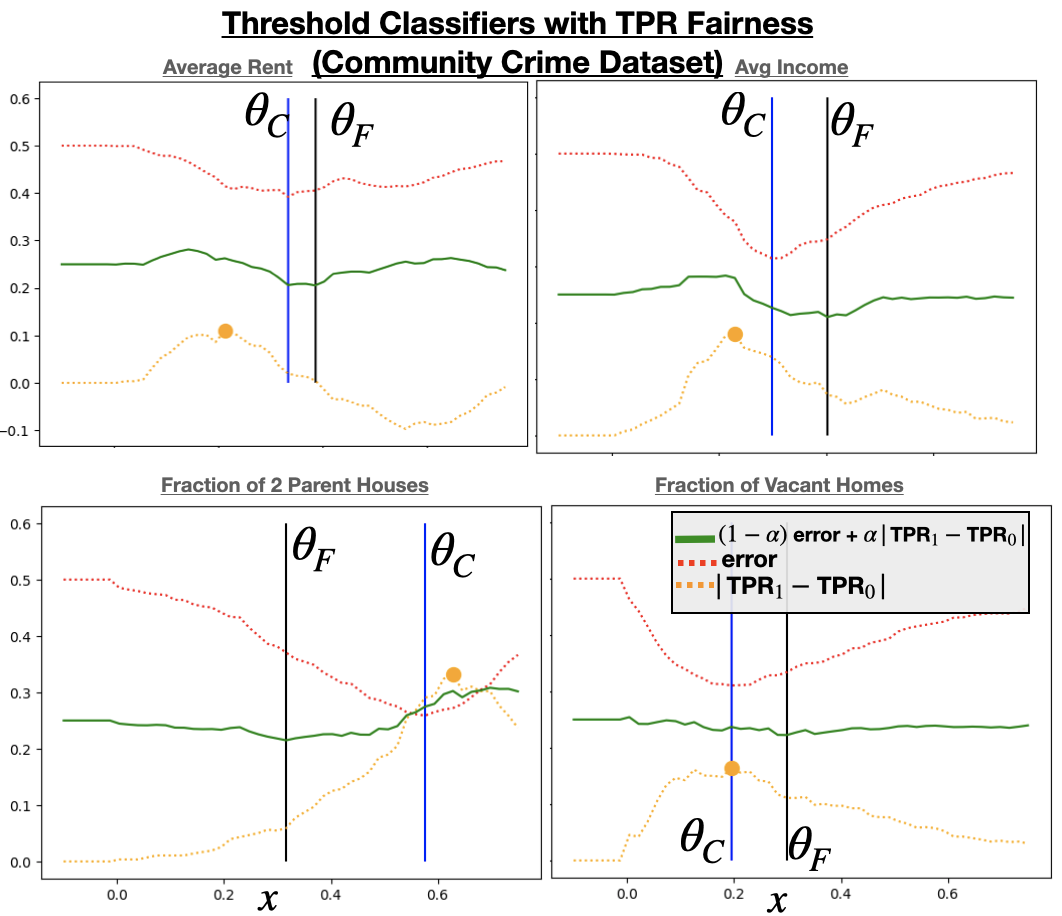}
    \caption{}
    \label{fig:single_var_9}
\end{figure}

\begin{figure}
    \centering
    \includegraphics[scale=0.35]{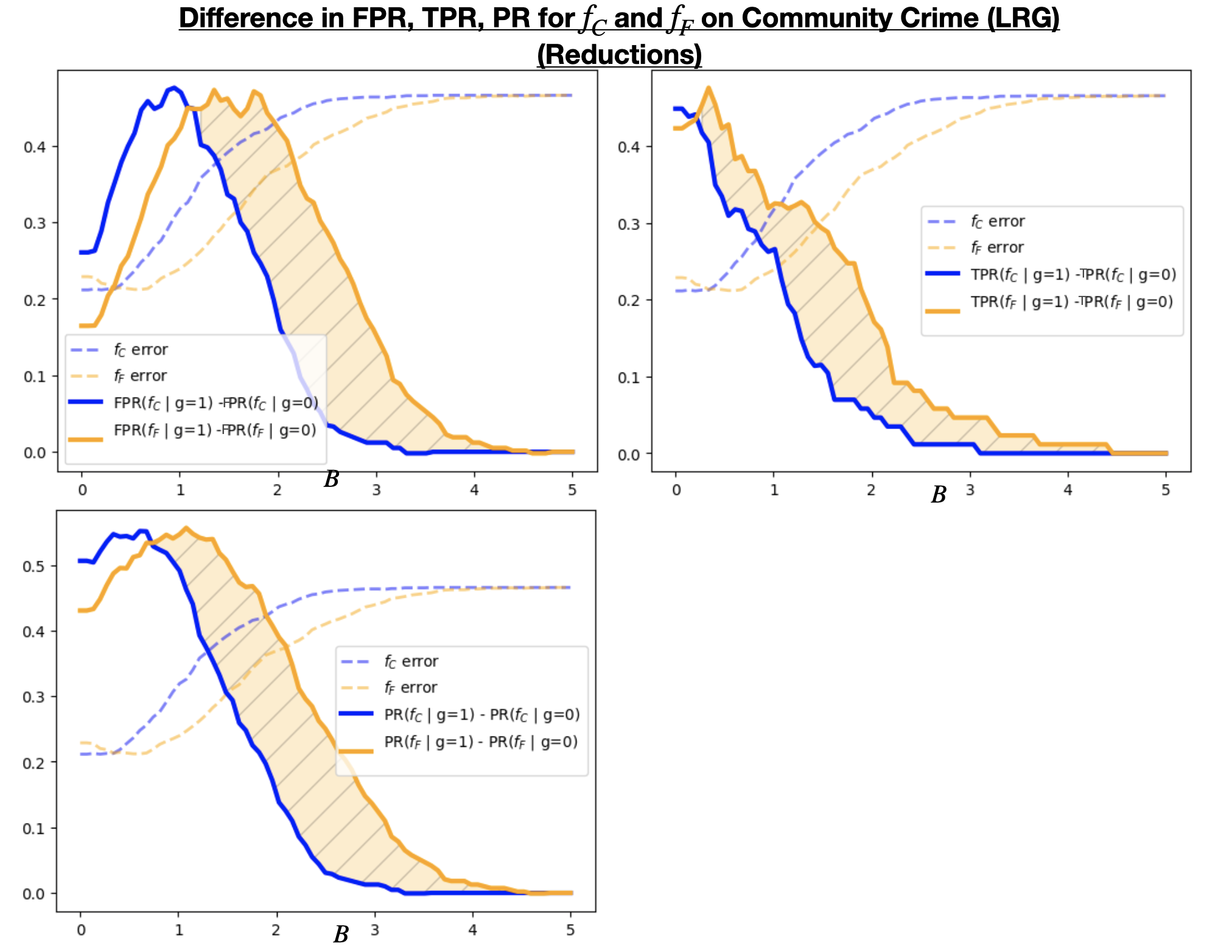}
    \caption{Difference in fairness for Logistic Regression (blue) and Reduction (orange) when agents are strategic under $l_2$ cost with budget $B$. The shaded region indicates instances in which the fair classifier is less fair than its baseline counterpart. Note that when the fairness of $f_F$ and $f_C$ are reversed, there accuracy is reversed as well. Once hitting their maximally unfair point both $f_C$ and $f_F$ tend to decreased in unfairness at equal rates. However, $f_F$ typically hits this maximally unfair point for a $B$ larger than $f_C$. }
    \label{fig:CC_LRG_red}
\end{figure}
\begin{figure}
    \centering
    \includegraphics[scale=0.4]{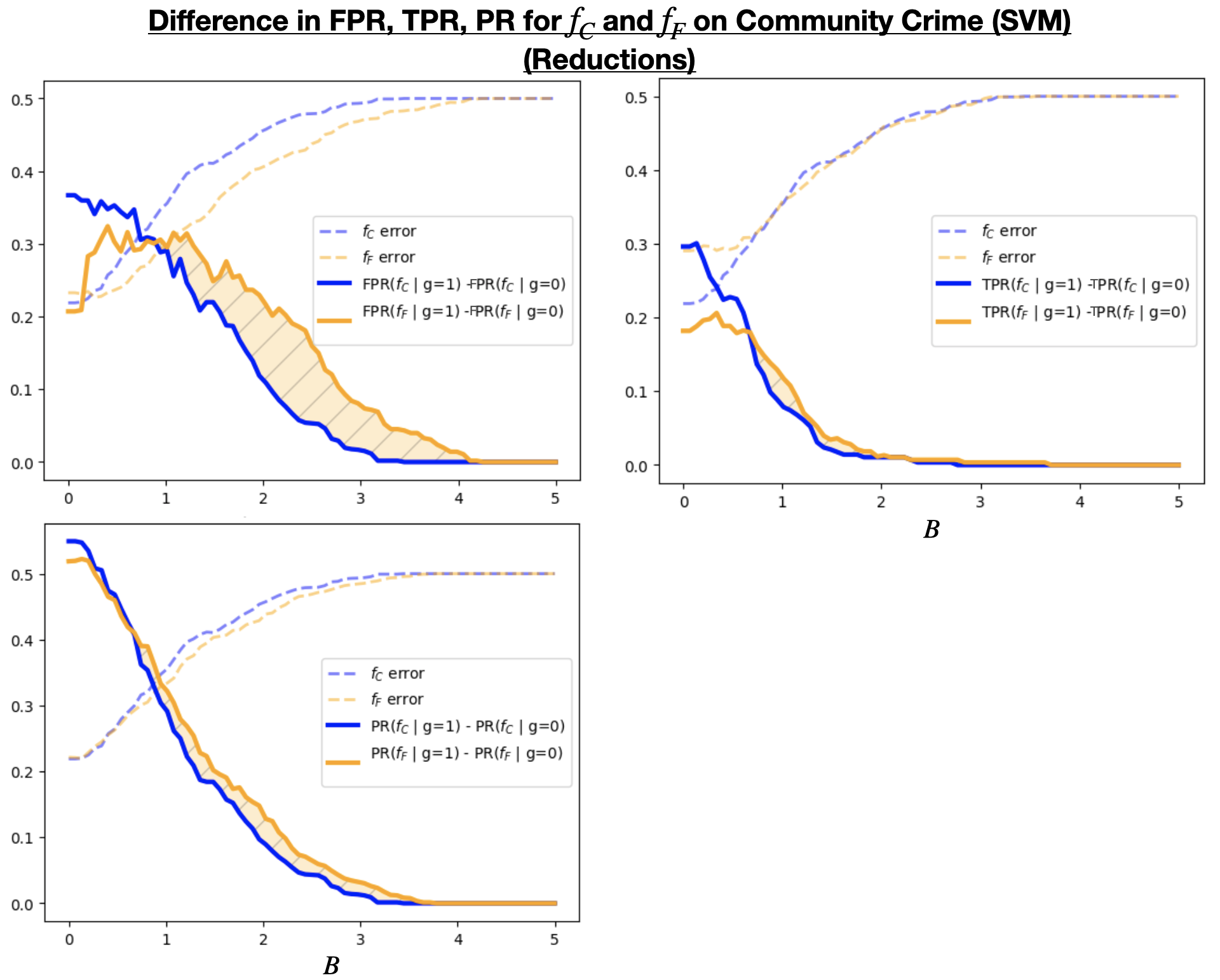}
    \caption{}
    \label{fig:CC_SVM_red}
\end{figure}
\begin{figure}
    \centering
    \includegraphics[scale=0.45]{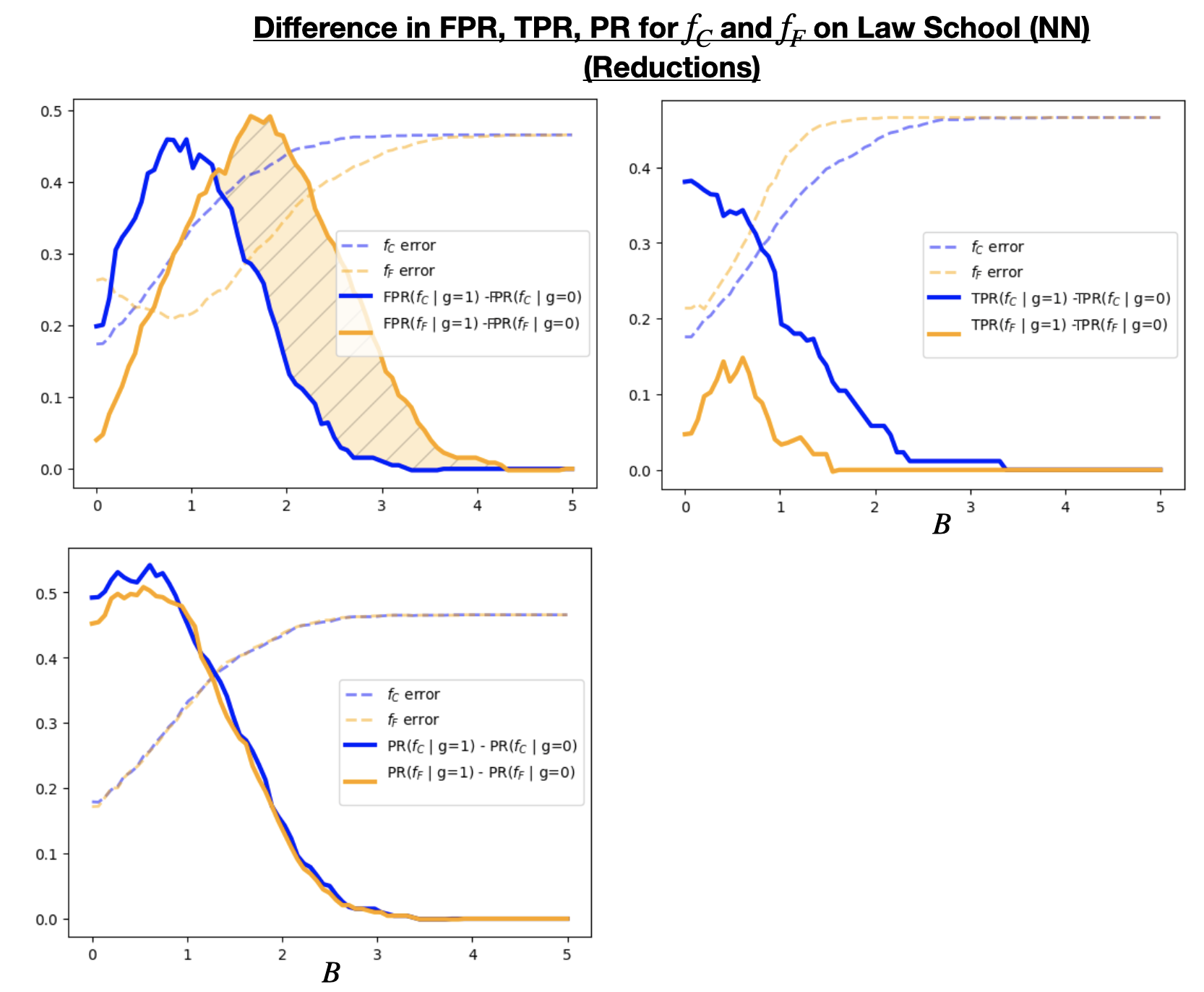}
    \caption{}
    \label{fig:LS_NN_red}
\end{figure}
\begin{figure}
    \centering
    \includegraphics[scale=0.45]{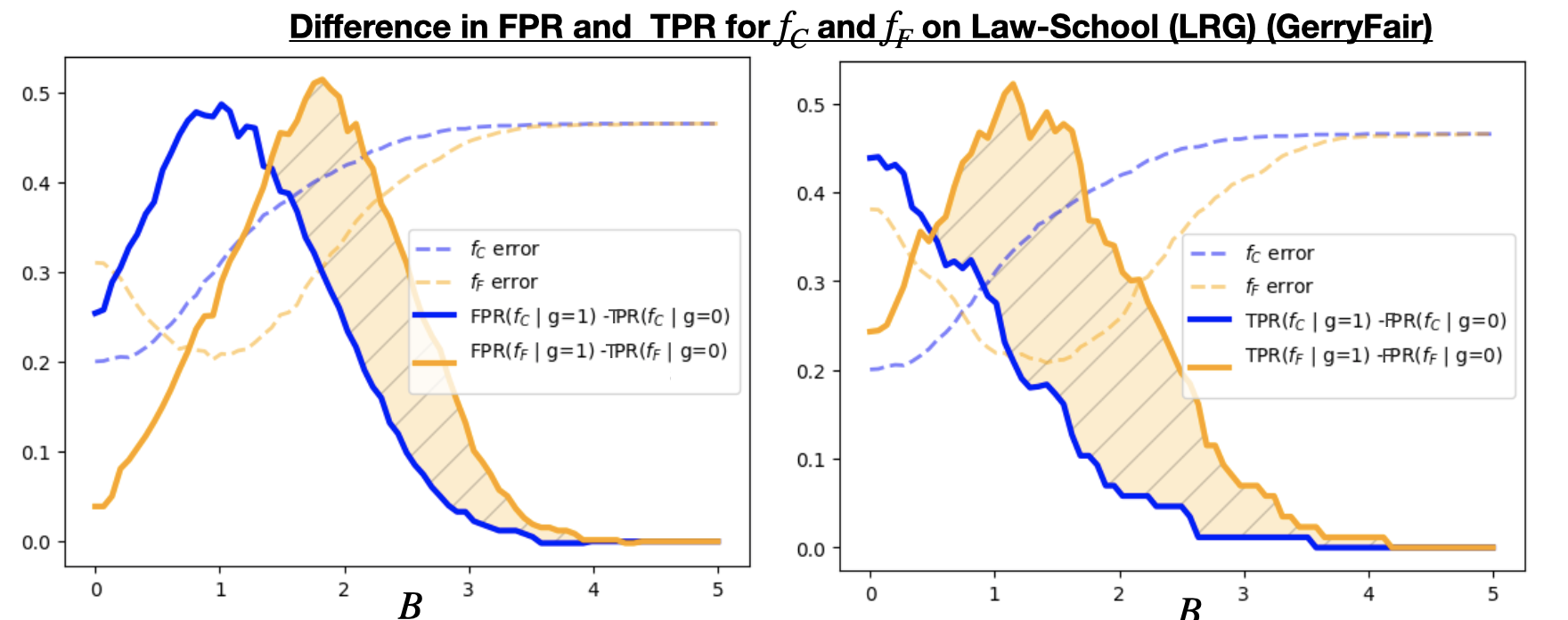}
    \caption{}
    \label{fig:LS_LRG_gerry}
\end{figure}
\begin{figure}
    \centering
    \includegraphics[scale=0.42]{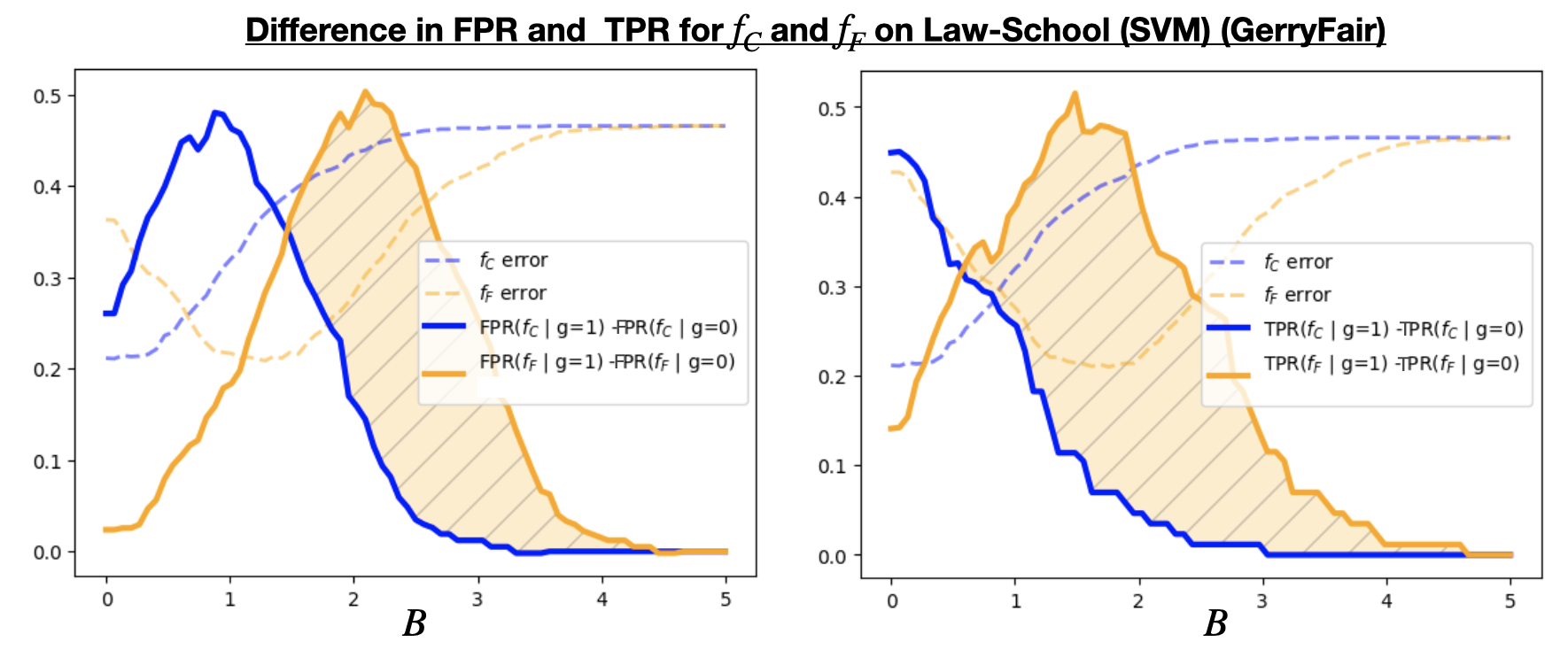}
    \caption{}
    \label{fig:LS_SVM_gerry}
\end{figure}

\begin{figure}
    \centering
    \includegraphics[scale=0.43]{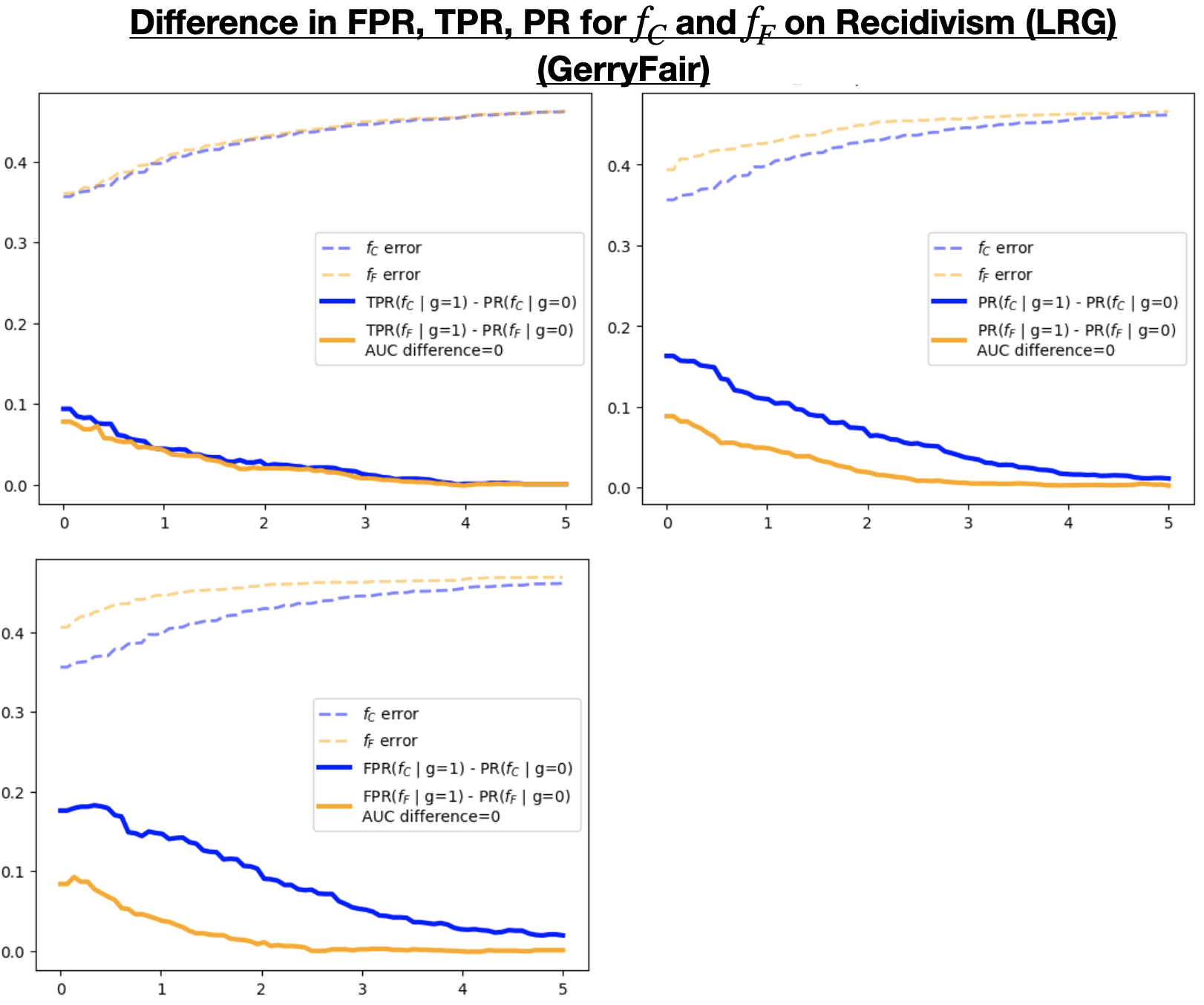}
    \caption{}
    \label{fig:redic_LRG_GF}
\end{figure}
\begin{figure}
    \centering
    \includegraphics[scale=0.45]{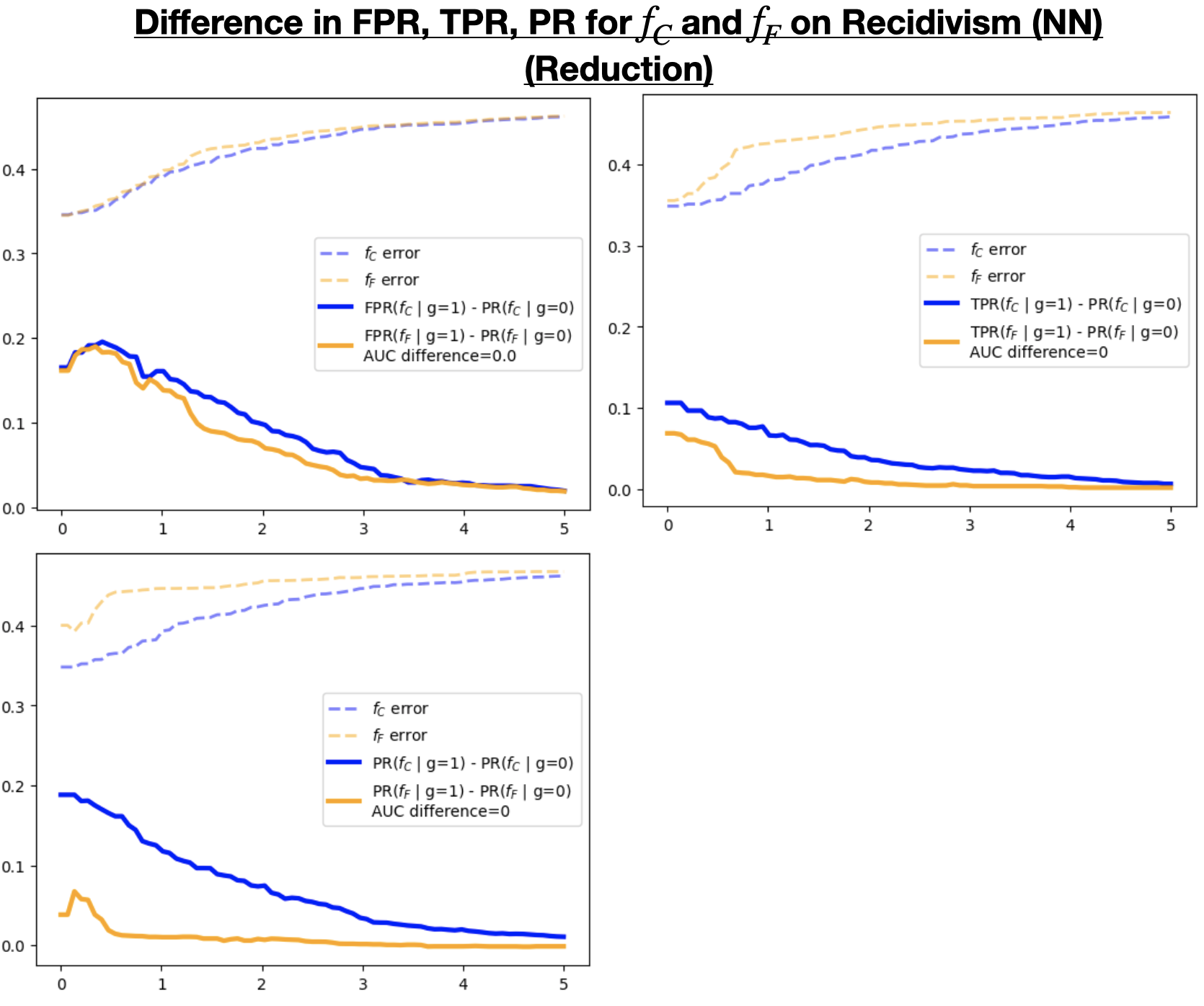}
    \caption{}
    \label{fig:multi_1}
\end{figure}
\begin{figure}
    \centering
    \includegraphics[scale=0.31]{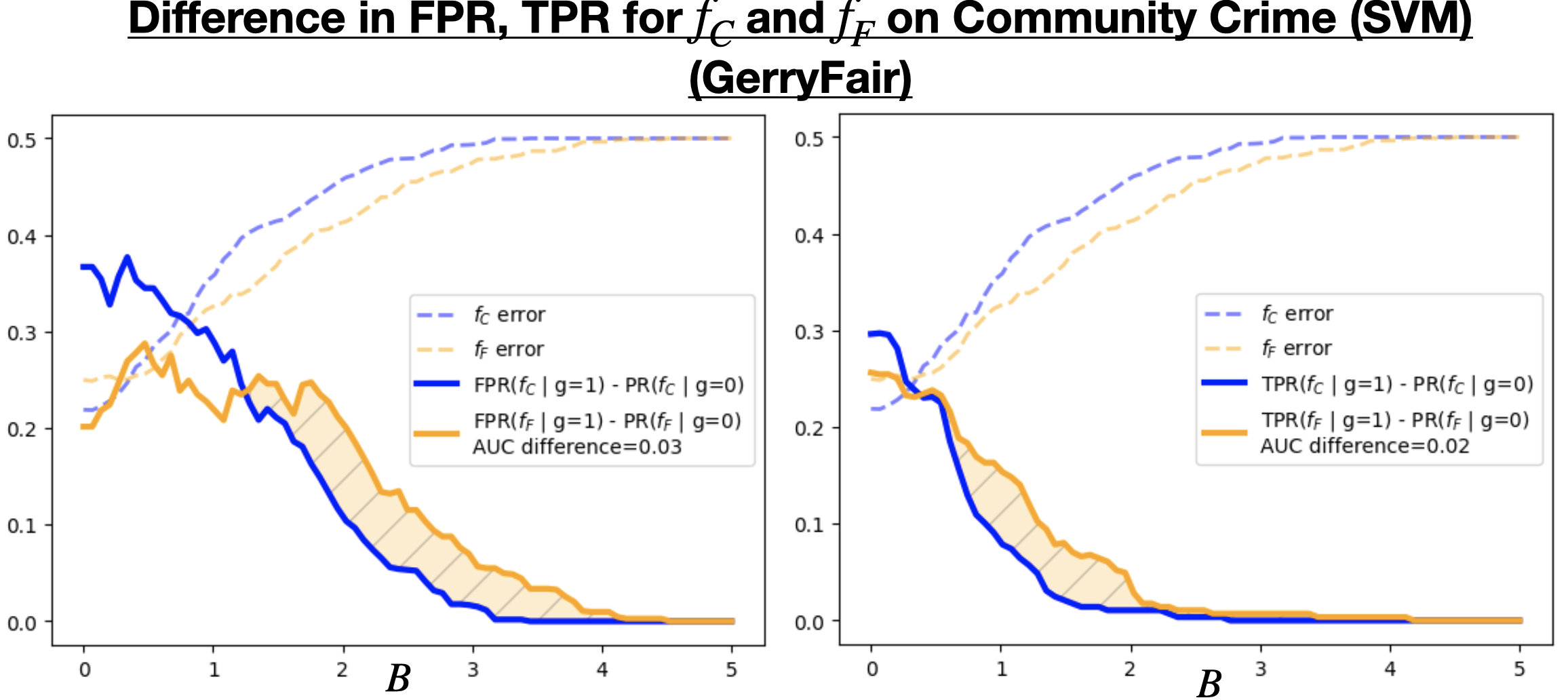}
    \caption{}
    \label{fig:multi_2}
\end{figure}

\end{document}